\documentclass{amsart}

\usepackage{amssymb,amsmath,amscd,amsthm,enumerate}
\usepackage{color}

\date{\today}

\newcommand{\A}{{\mathcal A}}
\newcommand{\Hi}{{\mathcal H}}
\newcommand{\Z}{{\mathbb Z}}
\newcommand{\R}{{\mathbb R}}
\newcommand{\C}{{\mathbb C}}

\newcommand{\D}{{\mathbb D}}

\newcommand{\E}{{\mathcal E}}

\newcommand{\Fib}{{\mathrm F}}

\newcommand{\TM}{{\mathrm{TM}}}

\DeclareMathOperator*{\slim}{s-lim}

\newtheorem{theorem}{Theorem}[section]
\newtheorem{lemma}[theorem]{Lemma}
\newtheorem{prop}[theorem]{Proposition}
\newtheorem{coro}[theorem]{Corollary}

\theoremstyle{definition}
\newtheorem{remark}[theorem]{Remark}
\theoremstyle{definition}

\theoremstyle{definition}
\newtheorem{ex}[theorem]{Example}

\allowdisplaybreaks

\numberwithin{equation}{section}

\renewcommand{\Re}{\mathrm{Re} \, }
\newcommand{\tr}{\mathrm{tr} }

\begin{document}

\title[Spreading Estimates for 1D Quantum Walks]{Spreading Estimates for Quantum Walks on the Integer Lattice via Power-Law Bounds on Transfer Matrices}


\author{David Damanik}
\thanks{D.\ D.\ was supported in part by NSF grants DMS--1067988 and DMS--1361625.}

\address{Department of Mathematics, Rice University, Houston, TX~77005, USA}

\email{damanik@rice.edu}

\author{Jake Fillman}

\thanks{J.\ F.\ was supported in part by NSF grants DMS--1067988 and DMS--1361625 (Corresponding Author)}

\address{Department of Mathematics, Rice University, Houston, TX~77005, USA}

\email{fillman@vt.edu}

\author{Darren C.\ Ong}

\address{Department of Mathematics, University of Oklahoma, Norman, OK~73019-3103, USA}

\email{darrenong@math.ou.edu}

\maketitle

\begin{abstract}

We discuss spreading estimates for dynamical systems given by the iteration of an extended CMV matrix. Using a connection due to Cantero--Gr\"unbaum--Moral--Vel\'azquez, this enables us to study spreading rates for quantum walks in one spatial dimension. We prove several general results which establish quantitative upper and lower bounds on the spreading of a quantum walk in terms of estimates on a pair of associated matrix cocycles. To demonstrate the power and utility of these methods, we apply them to several concrete cases of interest. In the case where the coins are distributed according to an element of the Fibonacci subshift, we are able to rather completely describe the dynamics in a particular asymptotic regime. As a pleasant consequence, this supplies the first concrete example of a quantum walk with anomalous transport, to the best of our knowledge. We also prove ballistic transport for a quantum walk whose coins are periodically distributed.

\end{abstract}



 \setcounter{tocdepth}{1}

\tableofcontents

\section{Introduction}

In recent years, quantum walks have been studied extensively; see \cite{AVWW, ABJ15, BGVW, CGMV, CGMV2, DFV, DMY2, GVWW, J11, J12, JM, K14, KS11, KS14, SK2010, ST12} for some papers on this subject that have appeared in the past five years. These are quantum analogues of classical random walks. For simplicity, let us focus on the important special case of walks on the integer lattice $\Z$.

A classical (stationary, or time-homogeneous) random walk on the integer lattice $\Z$ is given by transition probabilities $\{ p_{n,m} \}_{n,m \in \Z}$, where $p_{n,m} \in [0,1]$ denotes the probability of moving from site $m$ to site $n$. Denote the bi-infinite matrix with these entries by $P$. We therefore require that each column sum of $P$ is equal to one, that is, $\sum_{n \in \Z} p_{n,m} = 1$ for every $m \in \Z$. The state of the system may be specified by a column vector $v = (v_{n})_{n \in \Z}$ with entries $v_n \in [0,1]$, which is normalized in the sense that $\sum_{n \in \Z} v_n = 1$ so that $v_n$ describes the probability of the system being at site $n$. The time evolution of this state is then given by $v(k) = P^k v$ for $k \ge 0$ or, if $P$ is invertible, for $k \in \Z$.

An important special case is obtained by only considering nearest-neighbor transitions with symmetric transition probabilities. That is, we have $p_{n,m} = 0$ for $|n-m| > 1$ and $p_{n,m} = p_{m,n}$. In this case the matrix $P$ is tridiagonal and symmetric and hence has the structure of a Jacobi matrix. To emphasize this fact we will write in this case $J$ instead of $P$ for the transition matrix.\footnote{For a Jacobi matrix, one usually requires in addition that the off-diagonal terms are strictly positive. This would be in fact a natural additional assumption since the case $a_n := p_{n,n-1} = 0$ will correspond to two non-interacting walks on $\Z \cap (-\infty,n-1]$ and $\Z \cap [n,\infty)$, which can be studied separately.}
Thus, in this case, the evolution of an initial state $v$ is given by $v(k) = J^k v$. In particular, the probability of being at site $n$ at time $k$ is given by
\begin{equation}\label{e.Jacobievolution}
\langle \delta_n , J^k v \rangle,
\end{equation}
where $\delta_n$ is the vector with value $1$ at position $n$ and value $0$ otherwise, and $\langle v^{1}, v^{2} \rangle = \sum_{n \in \Z} \overline{v^{1}_n} v^{2}_n$. Here is where our symmetry assumption turns out to be critical. Since $J$ is obviously a bounded self-adjoint operator on $\ell^2(\Z)$, an expression like \eqref{e.Jacobievolution} can be rewritten with the help of the spectral theorem. Indeed, it takes the form $\int_\R x^k \, d\mu_{n,J,v}(x)$, where $\mu_{n,J,v}$ is a complex measure supported by the spectrum of $J$. Therefore, spectral theory enters the game and the long-time evolution of the classical random walk can be studied by means of a spectral analysis of the transition operator $J$.

To motivate the discussion of the quantum analogue, we can interpret the transitions as follows. Assuming further that the diagonal elements are zero as well, we have $p_{n-1,n} + p_{n+1,n} = 1$. Thus the transition from $n$ to one of its neighbors $n \pm 1$ corresponds to tossing a coin with probabilities $p_{n-1,n}$ and $1 - p_{n-1,n} = p_{n+1,n}$ and transitioning according to the outcome of this coin toss. A quantum walk on the integer lattice is described by assigning a spin to every site $n \in \Z$, which is given by a $\C^2$ vector, and determining the transitions from $n$ to one of its nearest neighbors $n \pm 1$ with the help of a ``quantum coin,'' which is a unitary $2 \times 2$ matrix. The details are described in Subsection~\ref{sec:qw} below. The upshot, following an important observation of Cantero, Gr\"unbaum, Moral, and Vel\'azquez in \cite{CGMV}, is the following. The time-evolution may now be described by an extended CMV matrix $\mathcal{E}$ in place of the Jacobi matrix $J$ above. CMV matrices are the natural unitary analogues of Jacobi matrices, and they play a canonical role within the class of unitary operators analogous to the canonical role played by Jacobi matrices within the class of (bounded) self-adjoint operators; see \cite{S1, S2} and references therein. Thus, we are now concerned with the study of
\begin{equation}\label{e.CMVevolution}
\langle \delta_n , \mathcal{E}^k v \rangle,
\end{equation}
where $\mathcal{E}$ is an extended CMV matrix, and in particular a unitary operator on $\ell^2(\Z)$. Again, this allows one to employ spectral theoretical methods to analyze the behavior of the corresponding quantum walk. In particular, one may relate continuity properties of spectral measures with respect to Hausdorff measures on the unit circle $\partial \D$ to spreading properties of the evolution \eqref{e.CMVevolution} and hence of the quantum walk; this approach was developed in \cite{DFV} and applied to concrete quantum walks in \cite{DMY2}.

The approach just mentioned works for arbitrary unitary operators. For extended CMV matrices, one can sometimes do better and obtain improved spreading estimates by using a different method. While the approach above is based on rewriting \eqref{e.CMVevolution} by means of integrals over the unit circle with respect to spectral measures, the alternative method we propose and develop in this paper rewrites \eqref{e.CMVevolution} by means of integrals over the unit circle with respect to normalized Lebesgue measure (i.e., the measure generated by normalized arc length). The fundamental formula is given in Lemma~\ref{l.parseval} below and was already observed in \cite{DFV}.\footnote{The formula was proved in \cite{DFV} in the hope that it would eventually become useful once one is able to work out the unitary analogues of a series of papers by Damanik and Tcheremchantsev in the self-adjoint case. The present paper realizes this vision.} It connects suitable averages of the quantities \eqref{e.CMVevolution} to integrals involving matrix elements of the resolvent of $\mathcal{E}$. This in turn connects time-averaged spreading to growth properties of transfer matrices, as the latter matrices are related to the properties of the resolvent probed by the integral in the fundamental formula. This connection is exceedingly useful in many cases of interest as there is a plethora of tools one can use to study the growth of transfer matrix norms. These tools are CMV analogues of tools originally developed in the Jacobi setting. In fact, some of these CMV versions are developed in the present paper.

This paper consists of two parts. The first part develops the general theory which enables one to deduce spreading estimates for the discrete time dynamics generated by an extended CMV matrix from bounds on transfer matrix norms for spectral parameters off the unit circle, which immediately gives general results when specializing to the case of quantum walks on the line (precise descriptions of the results may be found in Section~\ref{sec:results}). These general results are analogues for extended CMV matrices of results known for discrete Schr\"odinger operators. The latter operators generate a continuous time evolution via the Schr\"odinger equation. The connection between discrete Schr\"odinger operators, or more generally Jacobi matrices, and CMV matrices is very fruitful and has been explored extensively in recent years. Most of the time results are proved first in the self-adjoint context and are then carried over to the unitary context. There are some notable exceptions, such as Rakhmanov's theorem. Carrying over results from the self-adjoint case to the unitary case, or more concretely from the Jacobi to the CMV setting, is sometimes straightforward, sometimes prohibitively difficult, and sometimes doable but far from straightforward. What we do in this first part of the paper falls in the middle ground of doable but challenging. We obtain complete CMV analogues of the known Schr\"odinger results, but proving them required us to overcome quite a few obstacles. This is one of the cases where the ability to carry over the result should not be taken for granted but rather appreciated.

The second part presents various applications of the general theory. These applications range from simple but interesting observations to results that require a significant amount of work. Namely, in Sections~\ref{sec.6} and \ref{sec.8}, we will discuss lower bounds for the spreading of quantum walks associated with polymer models and the Thue-Morse substitution sequence. These results are relatively easy to establish and rely on the presence of a small set of exceptional spectral parameters at which the required estimates for the transfer matrix norms hold and are straightforward to verify. Nevertheless they are interesting because the spectral measures are either difficult to study (in the Thue-Morse case) or they are so singular that the approach from \cite{DFV} gives no non-trivial consequence for the spreading rates (for typical realizations of random polymer models). These examples demonstrate the usefulness of the method developed in this paper, which does not rely on spectral continuity properties and yet may be able to produce pretty strong estimates.

The applications that require a substantial amount of work consider the case of quantum walks with Fibonacci coins or with periodic coins. Quantum walks with coins generated by a Fibonacci substitution sequence, discussed in Section~\ref{sec:fqw}, are interesting because they can exhibit anomalous transport in the sense that the transport exponents (to be defined in the next section) take fractional values. To show this, one needs to establish upper and lower bounds on the spreading behavior of the quantum walk. This is another important advantage of the method developed here over the spectral continuity approach. The latter is incapable of establishing upper bounds purely in terms of spectral continuity (the case of random polymer models neatly demonstrates this). The method put forward here, however, naturally can be used to produce two-sided estimates, and in the Fibonacci case we work out in detail how this may be implemented. This requires quite extensive analysis (which is of dynamical and combinatorial nature) but in the end it produces rather sharp estimates, especially in the regime studied here. Our results rigorously establish transport behavior which was predicted numerically in the physics literature \cite{RMM}.

Finally, the case of quantum walks with periodic coins is clearly a fundamental one, for which one expects ballistic behavior to occur. That is, the walk spreads out in space with a well-defined strictly positive velocity. In Section~\ref{sec.9} we work out a result of this kind for periodic extended CMV matrices and derive the statement for quantum walks with periodic coins as a corollary.\footnote{During the review of this paper, it was kindly pointed out to us that the statement about ballistic transport for quantum walks on the line with periodic coins is a particular consequence of \cite[Theorem~4]{AVWW}. While the two proofs are related, we have decided to still present ours, which is inspired by \cite{AK98} and \cite{DLY}, in this paper for the convenience of the reader since the authors of \cite{AVWW} work in a more general setting and their paper does not state some particular consequences of their Theorem~4, such as our Theorem~\ref{t:ballistic} and its corollaries in Section~\ref{sec.9}.} While this happens only in the very last section of the paper, we want to emphasize that the proof is in fact independent of the bulk of the paper. Readers familiar with extended CMV matrices and the CGMV connection between quantum walks on the integer lattice and extended CMV matrices may skip ahead to Section~\ref{sec.9}.

We conclude the introduction with a brief remark on the half-line case. One-sided, or standard, CMV matrices may be used to study quantum walks on $\Z_+$, as shown also in \cite{CGMV}. The general theory we develop in part one of the paper has a natural half-line counterpart. In fact, some of the proofs are slightly simpler in this case. On the other hand, the applications we discuss in the second part of the paper all belong to the family of dynamically defined operators over invertible dynamics, and therefore are more naturally studied in the two-sided setting. One could study their half-line restrictions, but one would encounter possible discrete spectrum which changes the dynamical picture for initial states from the associated spectral subspace. To avoid cluttering the paper, we state and prove all theorems only in the two-sided case, and we describe how the theorems may be adapted to the one-sided setting.

\section*{Acknowledgements}

D.\ D.\ and J.\ F.\ would like to thank the Isaac Newton Institute for Mathematical Sciences, Cambridge, for support and hospitality during the programme ``Periodic and Ergodic Spectral Problems,'' where part of this work was undertaken. J.\ F.\ also thanks Hermann Schulz-Baldes for helpful comments on the literature, and Milivoje Lukic for helpful and stimulating discussions.
\newline

The authors are grateful to Axel Mondave for the translation of the abstract.

\part{General Results}

\section{Precise Statements} \label{sec:results}

\subsection{Unitary Dynamics}

We shall adopt the notation and development found in \cite{DFV,DT2010}. Consider a unitary operator $U$ acting on the separable Hilbert space $\Hi$, equipped with an orthonormal basis $\{\varphi_n\}_{n \in \Z}$. In this paper, we will primarily focus our attention on two closely related scenarios:
\begin{enumerate}

\item $U$ is an extended CMV matrix acting on $\Hi = \ell^2(\Z)$, equipped with the standard orthonormal basis $\varphi_n = \delta_n$,  $n \in \Z$.

\item $U$ is the one-step update rule of a quantum walk on $\Z$, and $\Hi = \ell^2(\Z) \otimes \C^2$. If we denote the standard basis of $\C^2$ by $\{e_\uparrow, e_\downarrow\}$, then the elementary tensors of the form
\begin{equation} \label{eq:orderedbasis}
\varphi_{2m} = \delta_m \otimes e_\uparrow, \;
\varphi_{2m+1} = \delta_m \otimes e_\downarrow, \;
m \in \Z
\end{equation}
comprise an orthonormal basis for $\Hi$.

\end{enumerate}

Given an intial state $\psi \in \Hi$ normalized by $\| \psi \| = 1$, we are interested in the time evolution of the vector $\psi$, that is, we want to study the evolution of $\psi(k) = U^k \psi$ as $k \in \Z_+$ grows. To quantify this, we first put
$$
a_{\psi}(n,k) = \left| \left\langle \varphi_n , \psi(k) \right\rangle \right|^2,
\quad
k \in \Z, \, n \in \Z,
$$
which can be thought of as the probability that $\psi$ is in the state $\varphi_n$ at time $k$. We shall also be interested in the time-averaged probabilities, given by
\begin{equation}\label{e.averagedprob}
\widetilde{a}_{\psi}(n,K) = \frac{2}{K} \sum_{k=0}^{\infty} e^{-2k/K} a_{\psi}(n,k),
\quad
K \in \Z_+, \, n \in \Z.
\end{equation}
Taking time averages of quantum dynamical quantities has a long tradition; compare, for example, the survey-type papers \cite{D16, DT2010, L96} and references therein. On a technical level, this has its roots in Wiener's theorem. In particular in the case of singular continuous spectral measures, Wiener's theorem allows one to prove rigorous estimates for time-averaged quantities, whereas in general similar estimates may not hold for non-time-averaged quantities. Additionally, the process of averaging can transform dynamical quantities in particularly pleasant ways. For example, the exponential averages considered in this paper recast the $\widetilde a$'s into averages of the Poisson kernel against spectral measures; this point of view is explored in more detail in \cite{DEHGV}.

To be completely proper, the normalizing factor in \eqref{e.averagedprob} should really be $\left(1-e^{-2/K}\right)$ so that $\widetilde a$ is also a probability distribution on $\Z$. Of course, $2/K$ is correct to first order in $K$, and we are primarily interested in quantities in the limit $K \to \infty$, so this distinction does not really matter. In particular, this does not affect the $\beta$'s (which will be defined below). We will be interested in such averages throughout the paper, so, for a function $f: \Z_{\geq 0} \to \R$, we introduce the notation $\langle f \rangle$ to denote the average of $f$. More precisely, we set
$$
\langle f \rangle(K)
=
\frac{2}{K} \sum_{k = 0}^{\infty} e^{-2k/K} f(k).
$$
For example, in this notation, one could write $\widetilde{a}_{\psi}(n,K) = \langle a_{\psi}(n,\cdot) \rangle(K)$.

Since $U$ is a unitary operator, $\| U^k \psi \| = 1$ for every $k$. We may then think of $U^k \psi $ as defining a probability distribution on $\Z$. From this point of view, it makes sense to describe the spreading of these distributions in terms of their moments. More precisely, for $p > 0$, $k \in \Z$, and $K \in \Z_+$, define
\begin{align*}
|X|^p_{\psi}(k)
& =
\sum_{n \in \Z} \left( |n|^p +1 \right) a_{\psi}(n,k), \\
\left\langle |X|^p_{\psi} \right\rangle(K)
& =
\sum_{n \in \Z} \left( |n|^p +1 \right) \widetilde{a}_{\psi}(n,K).
\end{align*}
We are interested in the spreading of states which are initially well-localized in the sense that $|X|^p_\psi(0) < +\infty$ for all $p > 0$. Consequently, we define
$$
\mathcal S
=
\mathcal S(\Hi)
:=
\left\{
\psi \in \Hi : \lim_{|n| \to \infty} |n|^p | \langle \varphi_n, \psi \rangle| = 0 \text{ for all } p > 0
\right\}.
$$
Obviously, $\mathcal S$ depends on the choice of basis, but we will not vary the orthonormal bases of $\ell^2(\Z)$ and $\ell^2(\Z) \otimes \C^2$ in this paper, so we will not make this dependence explicit in the notation.
\bigskip

\noindent \textit{Remark.}
It is helpful to observe that the probability of finding the wave packet outside of a ball of radius $R$ can be used to find a simple lower bound on the moments, viz.
\begin{align}
|X|^p_\psi(k)
& \geq
R^p \sum_{|n| \geq R} a_\psi(n,k),
\quad \forall R > 0, \, k \in \Z, \\
\label{eq:moments:pout}
\left\langle |X|^p_\psi \right\rangle(K)
& \geq
R^p \sum_{|n| \geq R} \widetilde a_\psi(n,K),
\quad \forall R > 0, \, K \in \Z_+.
\end{align}

In the two situations described above, there is a universal ballistic bound on the associated unitary dynamics; compare \cite[Theorem~2.22]{DT2010}. More precisely, if $U$ is an extended CMV matrix or the update rule of a quantum walk, $\psi \in \mathcal S$, and $p>0$, there is a constant $C = C_{\psi,p} > 0$ such that
$$
|X|^p_\psi(k)
\leq
Ck^p
\text{ for all } k \in \Z_+.
$$
Consequently, we would like to compare the growth of the $p$th moment to polynomial growth of the form $k^{\beta p}$ for a suitable exponent $\beta \in [0,1]$.  In light of this, the following transport exponents\footnote{Some authors consider Ces\`aro averages for the moments, instead of the exponential averages which we consider. However, it is not hard to see that either method of averaging yields the same values for $\widetilde\beta^\pm$; compare \cite[Lemma~2.19]{DT2010}} are natural objects to consider
\begin{align*}
\beta_{\psi}^+(p) & = \limsup_{k \to \infty} \frac{\log \left( |X|^p_{\psi}(k) \right)}{p \log(k)}, \\
\beta_{\psi}^-(p) & = \liminf_{k \to \infty} \frac{\log \left( |X|^p_{\psi}(k) \right)}{p \log(k)}, \\
\widetilde{\beta}_{\psi}^+(p) & = \limsup_{K \to \infty} \frac{\log \left( \left\langle |X|^p_{\psi} \right\rangle (K) \right)}{p \log(K)}, \\
\widetilde{\beta}_{\psi}^-(p) & = \liminf_{K \to \infty} \frac{\log \left( \left\langle|X|^p_{\psi}\right\rangle(K) \right)}{p \log(K)}.
\end{align*}
By Jensen's inequality, $\beta^{\pm}_{\psi}$ and $\widetilde{\beta}^{\pm}_{\psi}$ are all non-decreasing functions of $p$ \cite[Lemma~2.7]{DT2010}.

\subsection{CMV Matrices and the Szeg\H{o} Cocycle}
Given a sequence $\{ \alpha_n \}_{n \in \Z}$ of complex numbers where $\alpha_n \in \D = \{ z \in \C : |z| < 1\}$ for every $n \in \Z$, the associated \emph{extended CMV matrix}, $\E = \E_\alpha$, is a unitary operator on $\ell^2(\Z)$ defined by the matrix representation
\begin{equation} \label{def:extcmv}
\small
\E
=
\begin{pmatrix}
\ddots & \ddots & \ddots &&&&&  \\
\overline{\alpha_0}\rho_{-1} & -\overline{\alpha_0}\alpha_{-1} & \overline{\alpha_1}\rho_0 & \rho_1\rho_0 &&& & \\
\rho_0\rho_{-1} & -\rho_0\alpha_{-1} & -\overline{\alpha_1}\alpha_0 & -\rho_1 \alpha_0 &&& & \\
&  & \overline{\alpha_2}\rho_1 & -\overline{\alpha_2}\alpha_1 & \overline{\alpha_3} \rho_2 & \rho_3\rho_2 & & \\
& & \rho_2\rho_1 & -\rho_2\alpha_1 & -\overline{\alpha_3}\alpha_2 & -\rho_3\alpha_2 &  &  \\
& &&& \overline{\alpha_4} \rho_3 & -\overline{\alpha_4}\alpha_3 & \overline{\alpha_5}\rho_4 & \rho_5\rho_4 \\
& &&& \rho_4\rho_3 & -\rho_4\alpha_3 & -\overline{\alpha_5}\alpha_4 & -\rho_5 \alpha_4  \\
& &&&& \ddots & \ddots &  \ddots
\end{pmatrix},
\end{equation}
where $\rho_n = \left( 1 - |\alpha_n|^2 \right)^{1/2}$ for every $n \in \Z$. Since we heavily use the paper \cite{GZ06} in what follows, we adopt their convention for the location of the diagonal elements, that is, we have $\langle \delta_n, \E \delta_n \rangle = -\overline{\alpha_{n+1}} \alpha_n$ for $n \in \Z$ (so that our $\alpha_n$ is their $-\overline{\alpha_n}$). We refer to $\{\alpha_n\}_{n \in \Z}$ as the sequence of \emph{Verblunsky coefficients} of $\E$. Note that all unspecified matrix entries are implicitly assumed to be zero. As we have said in the introduction, one may also study (half-line) CMV matrices, which are given by setting $\alpha_1 = -1$ and restricting $\E$ to $\ell^2(\Z_+)$.

For $\alpha \in \D$ and $z \in \partial \D$, the corresponding Szeg\H{o} matrix is defined by
$$
S(\alpha,z)
=
\frac{1}{\rho} \begin{pmatrix}
z & - \overline{\alpha} \\
-\alpha z & 1
\end{pmatrix},
\quad
\rho = \sqrt{1 - |\alpha|^2}.
$$
The Szeg\H{o} transfer matrices associated to $\E_\alpha$ are then defined by
$$
T(n,m;z)
=
\begin{cases}
S(\alpha_{n-1},z) \cdots S(\alpha_m,z) & n > m \\
I & n = m \\
S(\alpha_n,z)^{-1} \cdots S(\alpha_{m-1},z)^{-1} & n < m
\end{cases}
$$
where $n, m \in \Z$. These matrices arise in the setting of half-line CMV matrices as propagation matrices of the orthonormal polynomials associated to the $\delta_0$ spectral measure. That is to say, $T(n,m;z)$ may be used to obtain the degree-$n$ orthonormal polynomials from the degree-$m$ orthonormal polynomials; see \cite{S1, S2} for a much more comprehensive discussion. The Szeg\H{o} transfer matrices are still quite relevant to the spectral analysis of extended CMV matrices; for example, the spectrum of an extended CMV matrix may be characterized via exponential growth properties of the Szeg\H{o} matrices \cite{DFLY2}.

Our first main result says that quantitative polynomial bounds on the Szeg\H{o} matrices imply quantitative lower bounds on the spreading of associated wave packets defined by iterating $\E$. The precise formulation follows.

\begin{theorem} \label{t:dt:powerlaw}
Let $\E = \E_\alpha$ be an extended CMV matrix whose Verblunsky coefficients are bounded away from $\partial \D$ in the sense that
\begin{equation} \label{alpha:bdawayfromcirc}
M
:=
\| \alpha \|_\infty
<
1.
\end{equation}
Suppose further that there exist $C > 0$, $\gamma \ge 0$,  such that the following condition holds: for every $R \ge 1$, there exists a nonempty Borel set $A_R \subseteq \partial \D$ such that
\begin{equation} \label{eq:szego:powerbounds}
\sup_{z \in A_R} \max_{|n|,|m| \le R}
\| T(n,m;z) \|
\leq
CR^\gamma.
\end{equation}
For each $K \in \Z_+$, let $R_K = K^{1/(1+\gamma)}$ and let $B_K$ denote the $1/K$-neighborhood of the set $A_{R_K}$ in $\partial \D$, i.e.,
$$
B_K
=
\left\{
z \in \partial \D : |z - z'| < \frac{1}{K} \text{ for some } z' \in A_{R_K}
\right\}.
$$
Under the previous assumptions, there exists $\widetilde C > 0$ such that the following bound holds for $\psi = \delta_0$:
\begin{equation} \label{eq:poutbound}
\sum_{|n| \geq R_K/3} \widetilde a_{\psi}(n,K)
\geq
\widetilde C |B_K| K^{\frac{-3\gamma}{1+\gamma}}
\text{ for all } K \in \Z_+,
\end{equation}
where $|\cdot|$ denotes normalized Lebesgue measure on $\partial \D$. Thus, for each $p > 0$, there exists a constant $\widetilde C_p > 0$ such that
\begin{equation} \label{eq:momentbound}
\left\langle |X|^p_{\psi} \right\rangle(K)
\geq
\widetilde C_p |B_K| K^{\frac{p-3\gamma}{1+\gamma}}
\text{ for all } K \in \Z_+.
\end{equation}
\end{theorem}

\begin{remark} \label{rem:shift} Let us make a few remarks about the statement of Theorem~\ref{t:dt:powerlaw}.
\begin{enumerate}
\item The assumptions of the theorem are translation-invariant. More precisely, suppose $\{\alpha_n\}_{n \in \Z}$ is such that there exist $C$, $\gamma$, and $A_R$ as in the statement of the theorem. Then, if $\alpha_n' = \alpha_{n+k}$ for some $k$ and all $n$, then there exists $C' > 0$ such that \eqref{eq:szego:powerbounds} holds (with the same choices for $\gamma$ and $A_R$). This is quite reasonable on physical grounds: global dynamical characteristics should not depend on how one chooses the origin.

\item The only step in the proof of Theorem~\ref{t:dt:powerlaw} which uses \eqref{alpha:bdawayfromcirc} is the step which passes between the Szeg\H{o} and Gesztesy--Zinchenko cocycles (the GZ cocycle will be defined below). In particular, if $\alpha$ is not bounded away from $\partial \D$, then Theorem~\ref{t:dt:powerlaw} still holds, but with $T$ replaced by $Z$.

\item Almost no modification is required to state and prove Theorem~\ref{t:dt:powerlaw} in the half-line setting. In particular, one need only replace \eqref{eq:szego:powerbounds} by
$$
\sup_{z \in A_R} \max_{0 \le n,m \le R}
\| T(n,m;z) \|
\leq
CR^\gamma.
$$
The conclusions of Theorem~\ref{t:dt:powerlaw} then hold verbatim in the half-line case. The modifications to the proof are straightforward and left to the reader.

\end{enumerate}
\end{remark}

Notice that Theorem~\ref{t:dt:powerlaw} supplies nontrivial dynamical bounds even when $A_R$ consists of a single point for large $R$. Indeed, this gives a simple direct proof of time-averaged transport which is nearly ballistic (in the large $p$ limit) for quantum walks in a Thue-Morse quasicrystal environment, and for quantum walks in a random polymer with critical spectral parameters, which we will describe in later sections. The latter example also shows that Theorem~\ref{t:dt:powerlaw} sometimes can give much more information than results which relate dynamics to fractal dimensions of spectral measures, such as \cite[Proposition~3.9]{DFV}. Specifically, in random polymer quantum walks, one expects spectral measures to be of pure point type (it should be noted that no rigorous proof of this fact exists). Thus, the results of \cite{DFV} give no information, while Theorem~\ref{t:dt:powerlaw} yields almost-ballistic transport (at least in the limit $p\to\infty$).

\subsection{The Gesztesy--Zinchenko Cocycle and Dynamical Upper Bounds}

In order to relate the Szeg\H{o} cocycle to quantum dynamical transport behavior, we often need to relate the Szeg\H{o} cocycle to solutions to the difference equation $\E u = zu$ with $z \in \C \setminus \{0\}$ and $u \in \C^{\Z}$. To that end, consider the  matrices
\begin{equation} \label{gz:onestepmats:def}
P(\alpha,z)
=
\frac{1}{\rho}
\begin{pmatrix}
-\overline\alpha & z \\
z^{-1} & - \alpha
\end{pmatrix}
,
\quad
Q(\alpha,z)
=
\frac{1}{\rho}
\begin{pmatrix}
-\alpha & 1 \\
1 & - \overline \alpha
\end{pmatrix},
\, \alpha \in \D, z \in \C \setminus \{0\},
\end{equation}
where $\rho = \left( 1 - |\alpha|^2\right)^{1/2}$ as before. Notice that
$$
\det(P(\alpha,z))
=
\det(Q(\alpha,z))
=
-1
\text{ for all } \alpha \in \D, \, z \in \C \setminus \{0\}.
$$
These matrices come from \cite{GZ06}, though we follow the convention of \cite{O14} and replace $\alpha$ by $-\overline{\alpha}$ in Gesztesy-Zinchenko's definition (the definition of a CMV matrix in \cite{GZ06} also replaces our $\alpha$ by $-\overline{\alpha}$). One may use $P$ and $Q$ to capture the recursion described by the difference equation $\E u = zu$ in the following sense. Any extended CMV matrix $\E$ enjoys a factorization $\E = \mathcal{L} \mathcal{M}$ \cite[Proposition~4.2.4]{S1}, where $\mathcal L$ and $\mathcal M$ can be decomposed as direct sums of $2 \times 2$ unitary matrices of the form
$$
\Theta(\alpha)
=
\begin{pmatrix}
\overline{\alpha} & \rho \\
\rho & - \alpha
\end{pmatrix}.
$$
Specifically,
\[
\mathcal L
=
\bigoplus_{j \in \Z} \Theta(\alpha_{2j}),
\quad
\mathcal M
=
\bigoplus_{j \in \Z} \Theta(\alpha_{2j-1}),
\]
where $\Theta(\alpha_n)$ acts on coordinates $n-1$ and $n$. If $u \in \C^{\Z}$ is such that $\E u = zu$, then we define $v := \mathcal L^{-1} u$. It follows that $\mathcal M u = z v$, and we have $\E^\top \! v = zv$. With $\Phi(n) := (u(n), v(n))^\top$, Gesztesy and Zinchenko show that
\begin{equation} \label{eq:gz:stepbystep}
\Phi(n)
=
\begin{cases}
P(\alpha_n,z)
\Phi(n-1)
& n \text{ is odd} \\
Q(\alpha_n,z)
\Phi(n-1)
& n \text{ is even}
\end{cases}
\end{equation}
for all $n \in \Z$; see also \cite[Proposition~3]{MSB} for a  generalization of this formalism to scattering zippers. This motivates the following definition. Denote $Y(n,z) = P(\alpha_n,z)$ when $n$ is odd and $Y(n,z) = Q(\alpha_n,z)$ when $n$ is even; then, the \emph{Gesztesy-Zinchenko cocycle} is defined by
\begin{equation} \label{gz:def}
Z(n,m;z)
=
\begin{cases}
Y(n - 1,z) \cdots Y(m, z) & n > m \\
I & n = m \\
Y(n, z)^{-1} \cdots Y(m - 1, z)^{-1} & n < m
\end{cases}
\end{equation}
If $u$, $v$, and $\Phi$ are as above, we have
\begin{equation} \label{eq:gzsoltransfer}
\Phi(n - 1)
=
Z(n, m; z) \Phi(m - 1)
\text{ for all } n,m \in \Z
\end{equation}
by \eqref{eq:gz:stepbystep}; compare \cite[Lemma~2.2]{GZ06}. We will often abbreviate the names and refer to $Z(n,m;z)$ as a GZ matrix.

By using the connection with solutions of the difference equation, (and hence with Green's functions), we can use the GZ matrices to bound the tails of the wave packet; compare \cite[Theorem~7]{DT07}.

\begin{theorem}\label{DT07.thm7}
Let $\E = \E_\alpha$ be an extended CMV matrix, consider the initial state $\psi = \delta_{-1}$, and define
\begin{align*}
\widetilde P_r(N,K)
& =
\sum_{n > N} \widetilde a_\psi(n,K),\\
\widetilde P_l(N,K)
& =
\sum_{n < -N} \widetilde a_\psi(n,K).
\end{align*}
If $\alpha$ is bounded away from $\partial \D$, there exists a constant $C_0 > 0$ such that
\begin{align}
\label{eq:pr:gzbound}
\widetilde P_r(N,K)
& \leq
C_0 K^3 \int_0^{2\pi} \left( \max_{1 \leq n \leq N} \left\Vert Z \left( n, 0; e^{1/K+i\theta} \right) \right\Vert^2\right)^{-1} \frac{d\theta}{2\pi}, \\
\label{eq:pl:gzbound}
\widetilde P_l(N,K)
& \leq
C_0 K^3 \int_0^{2\pi} \left( \max_{1 \leq -n \leq N} \left\Vert Z\left(n, 0; e^{1/K + i\theta} \right) \right\Vert^2\right)^{-1} \frac{d\theta}{2\pi}
\end{align}
for all $N, K \in \Z_+$.

\end{theorem}

We use $\psi = \delta_{-1}$ in Theorem~\ref{DT07.thm7} purely for convenience; the result also holds for any finitely supported initial state (the constant $C_0$ will depend on the choice of the initial state).

We can also use the GZ matrices to bound the tails without averaging, at the cost of an additional power of $k$.

\begin{theorem} \label{t:dt08}
Let $\E = \E_\alpha$ be an extended CMV matrix, consider the initial state $\psi = \delta_{-1}$, and define
\begin{align}
\label{pr:avbound}
P_r(N,K)
& =
\sum_{n > N} a_\psi(n,K),\\
P_l(N,K)
& =
\sum_{n < -N} a_\psi(n,K).
\end{align}
If $\alpha$ is bounded away from $\partial \D$, there exists a constant $C_0 > 0$ such that the left and right probabilities can be bounded as follows:
\begin{align}
\label{pr:bound}
P_r(N,k)
& \leq
C_0 k^4 \int_0^{2 \pi} \left( \max_{0 \leq n \leq N}
\left\| Z\left(n, 0; e^{1/k + i\theta} \right) \right\|^2 \right)^{-1} \, \frac{d\theta}{2\pi}
 \\
P_l(N,k)
& \leq
C_0 k^4 \int_0^{2 \pi} \left( \max_{-N \leq n \leq 0}
\left\| Z\left(n, 0; e^{1/k + i\theta} \right) \right\|^2 \right)^{-1} \, \frac{d\theta}{2\pi}
\end{align}
for all $N, k \in \Z_+$.
\end{theorem}

These bounds on the tails imply corresponding bounds on the transport exponents in terms of bounds on the GZ matrices (averaged over the spectral parameter).

\begin{theorem} \label{DT07.thm1}
Given $\E = \E_\alpha$, with $\alpha$ bounded away from $\partial \D$, suppose that there exist $C \in (0,\infty)$ and $\gamma\in (0,1)$ such that
\begin{equation}\label{RightInequality}
\sup_{K \in \Z_+} K^m \int_0^{2\pi} \! \left(\max_{0 \leq n \leq CK^\gamma}\left\Vert Z \left( n, 0; e^{1/K+i\theta} \right) \right\Vert^2 \right)^{-1} \frac{d\theta}{2\pi}
<
\infty,
\end{equation}
and
\begin{equation}\label{LeftInequality}
\sup_{K \in \Z_+} K^m \int_0^{2\pi} \! \left(\max_{1 \leq -n \leq CK^\gamma}\left\Vert Z \left( n, 0; e^{1/K+i\theta} \right) \right\Vert^2 \right)^{-1} \frac{d\theta}{2\pi}
<
\infty
\end{equation}
for all $m \ge 1$. Then $\beta^+_{\psi}(p) \leq \gamma$ for every $p > 0$ and every finitely supported $\psi \in \ell^2(\Z)$.
\end{theorem}

Let us say a few words to compare and contrast Theorem~\ref{t:dt:powerlaw} with Theorem~\ref{DT07.thm1}. First, there is a (very) rough physical heuristic which identifies the sizes of the transfer matrices with the sizes of the barriers that a wavepacket must tunnel through to escape compact regions. Thus, if all transfer matrices are sufficiently small, a wavepacket only observes small barriers, and hence, one expects it to tunnel effectively through all barriers; this vague statement is exactly what is made precise in the statement of Theorem~\ref{t:dt:powerlaw}. In contrast, if transfer matrices are large, then a wavepacket observes large barriers and has much more difficulty tunneling out of compact regions; this statement is made precise in Theorem~\ref{DT07.thm1}.

Moreover, notice that in contrast to Theorem~\ref{t:dt:powerlaw}, the application of Theorem~\ref{DT07.thm1} requires one to prove effective lower bounds on the transfer matrices on a set with nonzero Lebesgue measure. This contrasts with Theorem~\ref{t:dt:powerlaw} since as noted before, if one can prove effective upper bounds on transfer matrices at a single complex spectral parameter, Theorem~\ref{t:dt:powerlaw} already yields nontrivial conclusions! On physical grounds, this asymmetry in difficulty is unsurprising: to establish lower bounds on wavepacket propagation one need only prove that some portion of the wavepacket moves quickly, whereas to prove upper bounds one must simultaneously obtain effective control of the entire wavepacket.

\begin{remark} Let us conclude with remarks on a few extensions and variations on Theorems~\ref{DT07.thm7}, \ref{t:dt08}, and \ref{DT07.thm1}.
\begin{enumerate}
\item By general principles, $\widetilde \beta^+_\psi(p) \leq \beta^+_\psi(p)$ for any $\psi \in \mathcal S$ and any $p > 0$. Specifically, if $\gamma > \beta^+_\psi(p)$, then $|X|^p_\psi(k) \lesssim k^{\gamma p}$, and it is not hard to use this to show $\left\langle |X|^p_\psi \right\rangle(K) \lesssim K^{\gamma p}$; compare the proof of \cite[Lemma~7.2]{DLLY}, for example. In particular, the previous theorem immediately yields upper bounds on $\widetilde \beta^\pm_\psi(p)$ for finitely supported $\psi \in \ell^2(\Z)$.

\item Over the course of the proofs, we use the hypothesis that $\alpha$ is bounded away from $\partial \D$ in two places: to pass between $Z(N,0;z)$ and $Z(2\lfloor N/2 \rfloor, 0;z)$ in the proof of Lemma~\ref{DT07.Lemma3}, and, similarly, to pass between $Z(N,0;z)$ and $Z(4\lfloor N/4 \rfloor,0;z)$ in the proof of Lemma~\ref{DT07.Lemma3'}. Thus, one can modify the statements of Theorems~\ref{DT07.thm7}, \ref{t:dt08}, and \ref{DT07.thm1} to cover the case of general $\alpha$'s by restricting $n$ to be a multiple of 4 in the maxes which appear in Theorems~\ref{DT07.thm7}--\ref{DT07.thm1}.

\item In the half-line case, the ``left probabilities'' make no sense, since a state cannot propogate to the left of the origin. However, the bounds on the right probabilities (equations \eqref{pr:avbound} and \eqref{pr:bound}) still hold true in the half-line setting. In particular, the correct generalization of Theorem~\ref{DT07.thm1} to the half-line setting is the following: if \eqref{RightInequality} holds (for some $C$, some $\gamma$, and then all $m\ge1$), then $\beta_\psi^+(p) \le \gamma$ for all finitely supported $\psi$ and all $p > 0$.
\end{enumerate}
\end{remark}

\subsection{Motivation: Quantum Walks on $\Z$} \label{sec:qw}

We now precisely describe quantum walks on the integer lattice and their relationship with extended CMV matrices, following \cite{DFV}; see also \cite{CGMV}. A quantum walk is described by a unitary operator on the Hilbert space $\mathcal{H} = \ell^2(\Z) \otimes \C^2$, which models a state space in which a wave packet comes equipped with a ``spin'' at each integer site. As noted before, the elementary tensors of the form $\delta_n \otimes e_\uparrow$, and $\delta_n \otimes e_\downarrow$ with $n \in \Z$ comprise an orthonormal basis of $\Hi$. A time-homogeneous quantum walk scenario is given as soon as coins
\begin{equation}\label{e.timehomocoins}
Q_{n}
=
\begin{pmatrix}
q^{11}_{n} & q^{12}_{n} \\
q^{21}_{n} & q^{22}_{n}
\end{pmatrix}
\in \mathrm U(2), \quad n \in \Z,
\end{equation}
are specified. To avoid degenerate decoupling situations, we will always assume that $q_n^{11}, q_n^{22} \neq 0$. As one passes from time $t$ to time $t+1$, the update rule of the quantum walk is as follows,
\begin{align}
\delta_{n} \otimes e_\uparrow & \mapsto
  q^{11}_{n} \delta_{n+1} \otimes e_\uparrow
+ q^{21}_{n} \delta_{n-1} \otimes e_\downarrow , \label{e.updaterule1} \\
\delta_n \otimes e_\downarrow  & \mapsto
  q^{12}_{n} \delta_{n+1} \otimes e_\uparrow
+ q^{22}_{n} \delta_{n-1} \otimes e_\downarrow \label{e.updaterule2}.
\end{align}
If we extend this by linearity and continuity to general elements of $\mathcal{H}$, this defines a unitary operator $U$ on $\mathcal{H}$. Next, order the basis of $\mathcal{H}$ as in \eqref{eq:orderedbasis}, i.e. $\varphi_{2m} = \delta_m \otimes e_\uparrow$, $\varphi_{2m+1} = \delta_m \otimes e_\downarrow$ for $m \in \Z$. In this ordered basis, the matrix representation of $U : \mathcal{H} \to \mathcal{H}$ is given by
\begin{equation}\label{e.umatrixrep}
U
=
\begin{pmatrix}
\ddots &\ddots & \ddots & \ddots &&&&&&  \\
& 0 & 0 & q_0^{21} & q_0^{22} &&&&& \\
&q_{-1}^{11} & q_{-1}^{12} & 0 & 0 &&& && \\
&& & 0 & 0 & q_1^{21} & q_1^{22} & && \\
&& & q_0^{11} & q_0^{12} & 0 & 0 & && \\
&& &&& 0 & 0 & q_2^{21} & q_2^{22}& \\
&& &&& q_1^{11} & q_1^{12} & 0 & 0 & \\
&& &&&& \ddots & \ddots &  \ddots & \ddots
\end{pmatrix},
\end{equation}
which can be checked readily using the update rule \eqref{e.updaterule1}--\eqref{e.updaterule2}; compare \cite[Section~4]{CGMV}.

We can connect quantum walks to CMV matrices using the following observation. If all Verblunsky coefficients with even index vanish, the extended CMV matrix in \eqref{def:extcmv} becomes
\begin{equation}\label{e.ecmvoddzero}
\mathcal{E}
=
\begin{pmatrix}
\ddots & \ddots & \ddots & \ddots &&&&&& \\
& 0 & 0 & \overline{\alpha_1}  & \rho_1 &&&&& \\
& \rho_{-1} & -\alpha_{-1} & 0 & 0 &&&&& \\
&&  & 0 & 0 & \overline{\alpha_3}  & \rho_3 &&& \\
&& & \rho_1 & - \alpha_1 & 0 & 0 &  &&  \\
&& &&& 0 & 0 & \overline{\alpha_5}  & \rho_5 & \\
&& &&& \rho_3 & - \alpha_3 & 0 & 0  & \\
&& &&&& \ddots & \ddots &  \ddots & \ddots
\end{pmatrix}.
\end{equation}
The matrix in \eqref{e.ecmvoddzero} strongly resembles the matrix representation of $U$ in \eqref{e.umatrixrep}. Note, however, that $\rho_n > 0$ for all $n$, so \eqref{e.umatrixrep} and \eqref{e.ecmvoddzero} may not match exactly when $q_n^{kk}$ is not real and positive. However, this can be easily resolved by conjugation with a suitable diagonal unitary, as shown in \cite{CGMV}. Concretely, given $U$ as in \eqref{e.umatrixrep}, write
$$
q_n^{kk}
=
r_n \omega_n^k,
\quad
n \in \Z, \; k \in \{ 1,2 \}, \; r_n > 0, \; |\omega_n^k| = 1,
$$
define $\{ \lambda_n \}_{n \in \Z}$ by
$$
\lambda_0 = 1, \; \lambda_{-1} = 1, \; \lambda_{2n+2} = \omega_n^1 \lambda_{2n}, \; \lambda_{2n+1} = \overline{\omega^2_n} \lambda_{2n-1},
$$
and let $\Lambda = \mathrm{diag}(\ldots, \lambda_{-1} , \lambda_0 , \lambda_1 , \ldots)$. More precisely $\Lambda:\ell^2(\Z) \to \ell^2(\Z) \otimes \C^2$ maps $\delta_{n}$ to $\lambda_{n} \varphi_n$ for each $n \in \Z$ (where $\varphi_n$ is as in \eqref{eq:orderedbasis}). After a short calculation, we see that
$$
\mathcal{E} = \Lambda^* U \Lambda,
$$
where $\mathcal{E}$ is the extended CMV matrix corresponding to the Verblunsky coefficients
\begin{equation}\label{e.correspondence}
\alpha_{2n} = 0 , \quad
\alpha_{2n+1}
=
\frac{\lambda_{2n-1}}{\lambda_{2n}} \cdot \overline{q_n^{21}}
= - \frac{\lambda_{2n+1}}{\lambda_{2n+2}} q_n^{12}, \quad n \in \Z,
\end{equation}
where equivalence of the two expressions for $\alpha_{2n+1}$ follows from unitarity of $Q_n$. Notice that the hypotheses $q_n^{11}, q_n^{22} \neq 0$ imply $\rho_n > 0$ for all $n \in \Z$.

\bigskip

 The reader should be warned that many authors (e.g.\ \cite{CGMV,DFV}) use the \emph{transpose} of the matrix $U$ in \eqref{e.umatrixrep}. However, since we imagine a vector $\psi$ as a column vector with $U$ acting on the left, our representation is more appropriate for the present paper. For quantum dynamical purposes, this is not a major concern for the following reason. If $\E$ is an extended CMV matrix, then $\E^\top = S^* \E' S$, where $S$ denotes the left shift on $\ell^2(\Z)$ and $\E'$ has Verblunsky coefficients $\alpha_n' = \alpha_{n+1}$. Since all of our theorems are stable under shifting the Verblunsky coefficients, it does not matter whether one considers $U$ or $U^\top$ in what follows.

\bigskip

Let us conclude with a few words on half-line quantum walks. A quantum walk on $\Z_+$ has state space $\Hi_+ = \left( \ell^2(\Z_+) \otimes \C^2 \right) \oplus \langle \delta_0 \otimes e_\downarrow \rangle$, and is specified by coins $Q_n \in \mathrm U(2)$ for $n \ge 1$. Additionally, one specifies a reflecting boundary condition at the origin, e.g.,
$$
Q_0
=
\begin{pmatrix} 0 & 1 \\ -1 & 0 \end{pmatrix}.
$$
With this choice of boundary condition, one can check that \eqref{e.updaterule1} and \eqref{e.updaterule2} define a unitary operator on $\Hi_+$\footnote{In particular, the choice of $Q_0$ implies that $\delta_0 \otimes e_\downarrow$ is mapped to $\delta_1 \otimes e_\uparrow$ so that $U$ maps $\Hi_+$ to itself.}, and that this update rule has the form of a half-line CMV matrix with respect to the ordered basis $\{ \varphi_n : n \in \Z_+ \}$, where $\varphi_n$ is as in \eqref{eq:orderedbasis}.

\section{Power-Law Bounds on Szeg\H{o} Cocycles}

In this section, we will prove Theorem~\ref{t:dt:powerlaw}. We begin by collecting a few useful facts. The following variant of the Parseval formula allows us to connect time-averages of observables to averages of the resolvent over the spectral parameter.

\begin{lemma}\label{l.parseval}
Suppose $U$ is a unitary operator on $\Hi$, $\psi \in \Hi$, and $\{\varphi_n\}$ is an orthonormal basis of $\Hi$. For all $K \in \Z_+$ and $n$, we have
$$
\sum_{k \ge 0} e^{-2k/K} a_\psi(n,k)
=
e^{2/K} \int_0^{2\pi} \left| \left\langle \varphi_n , \left(U - e^{1/K + i\theta} \right)^{-1} \psi \right\rangle \right|^2 \, \frac{d\theta}{2\pi}.
$$
\end{lemma}

\begin{proof}

This is  \cite[Lemma~3.16]{DFV}.

\end{proof}

Next, we will need a Gronwall-type variational estimate for the Szeg\H{o} matrices. The key ingredient in this estimate is furnished by the following useful variational lemma. Throughout the remainder of this section, let $\alpha$ be a fixed sequence of Verblunsky coefficients which is bounded away from $\partial \D$.

\begin{lemma}
Let $z \in \partial \D$ and $w \in \C$. Then, for $n \ge 0$, we have
\begin{equation}\label{e.energyvariation1}
T(n,0;w)
=
T(n,0;z) - \sum_{\ell = 0}^{n-1} \left( 1 - z^{-1} w \right) T(n,\ell;z) P T(\ell,0;w),
\end{equation}
where
$$
P = \begin{pmatrix} 1 & 0 \\ 0 & 0 \end{pmatrix}.
$$
\end{lemma}

\begin{proof}
This is \cite[Proposition~10.8.1]{S2}.\footnote{Actually, \cite[Proposition~10.8.1]{S2} makes this statement only for $w \in \D$, but the proof given there in fact establishes it for all $w \in \C$. Alternatively, once one has \eqref{e.energyvariation1} for all $w \in \D$, the result for arbitrary $w \in \C$ follows immediately from analyticity.}
\end{proof}

The following bound is an OPUC analog of \cite[Lemma~2.1]{DT03}, which is the corresponding result for Schr\"odinger cocycles.

\begin{lemma} \label{l:energyvar}
Suppose $z \in \partial \D$ and $N \ge 1$. Denote
$$
L(N) = \max \{ \| T(n,m;z) \| : |n|, |m| \le N \}.
$$
Then, for $\delta \in \C$ with $|\delta| \leq 1$, we have
\begin{equation}\label{e.energyvariation2}
\left\| T\left(n,0;z e^\delta\right) \right\|
\le
L(N) \exp \big( 2 L(N) \cdot |n| \cdot |\delta| \big)
\end{equation}
for all $n$ with $|n| \leq N$.
\end{lemma}

\begin{proof}
Let us consider the case $0 \leq n \leq N$; the proof for $-N \le n < 0$ is similar. By \eqref{e.energyvariation1}, we have
$$
T\left(n,0;z e^\delta\right)
=
T(n,0;z) - \sum_{\ell = 0}^{n-1} \left( 1 - e^\delta \right) T(n,\ell;z) P T\left(\ell,0;z e^\delta\right).
$$
Now apply \eqref{e.energyvariation1} to each $T\left(\ell,0;z e^\delta\right)$ and iterate this procedure until all matrices correspond to spectral parameter $z$. Then, we obtain the estimate
\begin{align*}
\left\| T\left(n,0;z e^\delta\right) \right\| & \le \sum_{k = 0}^n \binom{n}{k} L(N)^{k+1} (2|\delta|)^k \\
& = L(N) ( 1 + 2 L(N) |\delta| )^n \\
& \le L(N) \exp \big( 2 L(N) n |\delta| \big),
\end{align*}
where we have used $\|P\| = 1$ and that $|1 - e^\delta| \le 2 |\delta|$ for $\delta \in \overline{\D}$.
\end{proof}

With these tools, Theorem~\ref{t:dt:powerlaw} can be deduced along the lines of \cite{DT03}. As noted in \cite{DFLY2}, the Gesztesy-Zinchenko matrices are related to the Szeg\H{o} matrices in a simple fashion. One has
\begin{equation}\label{eq:szego:gz:rel}
z^{-1} S(\alpha,z) S(\beta,z)
=
D_z^{-1} P(\alpha,z) Q(\beta,z) D_z
\end{equation}
for all $\alpha,\beta \in \D$ and $z \in \C \setminus \{0\}$, where
\[
D_z
=
\begin{pmatrix}
z & 0 \\
0 & 1
\end{pmatrix}
\]
We warn the reader: this formula does not exactly match the corresponding identity in \cite{DFLY2}; this happens because \cite{DFLY2} follows the conventions of \cite{S1,S2} with regards to the locations of the matrix elements of $\E$, while we follow the conventions in \cite{GZ06}. Additionally, we take $v = \mathcal L^{-1} u$ here, while \cite{DFLY2} takes $v = \mathcal M u$. With this relationship in hand, we can easily relate norm estimates on the two cocycles.

\begin{proof}[Proof of Theorem~\ref{t:dt:powerlaw}]
Since $M = \sup_n |\alpha_n| < 1$, we see that
\begin{equation} \label{eq:rhobound}
\sup_{n \in \Z} \rho_n
\leq
(1 - M^2)^{-1/2}
<
\infty.
\end{equation}
Using this and \eqref{eq:szego:gz:rel}, we see that there is a constant $C_0 > 0$ (which depends only on $M = \|\alpha\|_\infty$) such that
$$
\left\| Z \left(n, m ; ze^{1/K} \right) \right\|
\leq
C_0 \left\| T\left(n, m; ze^{1/K} \right) \right\|
$$
for all $n, \, m \in \Z$, all $K \in \Z_+$, and every $z \in \partial \D$. Fix $K \in \Z_+$, and let $z \in B_K$ be given. By definition, there exists $z' \in A_{R_K}$ with $|z-z'| < 1/K$, and, by assumption, we have
$$
L
=
L(R_K,z')
:=
\sup_{|n|,|m| \leq R_K} \| T(n,m;z') \|
\leq
CR_K^\gamma
$$
By Lemma~\ref{l:energyvar}, we have the following for each $n \in \Z$ such that $1 \leq n \leq R_K$:
\begin{align*}
\left\| T\left(n, 0; ze^{1/K} \right) \right\|
& \leq
L \exp\left(2L \cdot n \cdot \frac{2}{K} \right) \\
& \leq
C R_K^\gamma \exp\left( \frac{4 C R_K^{\gamma+1}}{K} \right),
\end{align*}
where we have used the power-law bound \eqref{eq:szego:powerbounds} in the second line. Since $R_K^{\gamma + 1} = K$, this yields
$$
\left\| Z \left(n, 0; ze^{1/K}\right) \right\|
\leq
DR_K^{\gamma}
\text{ for all }  1 \le n \le R_K,
$$
where $D = C_0 C e^{4 C}$. By \eqref{eq:rhobound},  we may enlarge $D$ in a $(K,z)$-independent fashion to get
\begin{equation} \label{eq:szego:energyvarest1}
\left\| Z \left(n, j; ze^{1/K}\right) \right\|
\leq
DR_K^{\gamma}
\text{ for all }  1 \le n \le R_K, \, 0 \le j \le 3.
\end{equation}
Now, take $\psi = \delta_0$, and define a pair of $\ell^2$ sequences $u$ and $v$ by $ u := \left(\E - ze^{1/K}\right)^{-1} \psi$ and $v := \mathcal L^{-1} u$; put $\Phi(n) = (u(n),v(n))^\top$, as before. Throughout the remainder of the argument $C_1, \, C_2, \ldots$ will denote constants which depend solely on $\E$ (and neither on $K$ nor $z$). Using \eqref{eq:gzsoltransfer} and \eqref{eq:szego:energyvarest1}, we have
\begin{equation} \label{eq:gzlowerbound}
\|\Phi(n)\|^2
\geq
C_1 R_K^{-2\gamma} \sum_{j=-1}^2 \| \Phi(j) \|^2
\text{ for all $1 \leq n < R_K$}.
\end{equation}
Notice that this uses $\det(Z) = \pm 1$ to get $\|Z^{-1}\| = \|Z\|$ by the singular value decomposition. Since $(\E - ze^{1/K}) u = \psi$, we have
$$
  \overline{\alpha_0} \rho_{-1} u_{-1}
- \overline{\alpha_0} \alpha_{-1} u_0
+ \overline{\alpha_1} \rho_0 u_1
+ \rho_1 \rho_0 u_2
- ze^{1/K} u_0
=
\psi_0
=
1.
$$
Consequently, we have
$$
\sum_{j=-1}^2 \| \Phi(j) \|^2
\ge
C_2.
$$
Summing \eqref{eq:gzlowerbound} over $2R_K / 3 < n < R_K$, we obtain the lower bound
$$
\sum_{n \geq 2R_K/3} \|\Phi(n)\|^2
\geq
\sum_{R_K \geq n \geq 2R_K/3} \|\Phi(n)\|^2
\geq
C_3 R_K^{1-2\gamma}.
$$
Since $v = \mathcal L^{-1} u$, we may use the explicit form of $\mathcal L$ to see that
$$
|v(n)|^2
\leq
C_4 \sum_{j=-1}^1 |u(n+j)|^2
\text{ for all } n \in \Z.
$$
Using this together with the definition of $u$, we obtain the following lower bound:
\begin{equation} \label{eq:pout:lb}
\sum_{n\geq R_K/3} \left| \left\langle \delta_n, \left( \E - ze^{1/K} \right)^{-1} \psi \right\rangle \right|^2
=
\sum_{n \ge R_K/3} |u(n)|^2
\geq
C_5 R_K^{1-2\gamma}.
\end{equation}
By Lemma~\ref{l.parseval}, we have
\begin{align*}
\widetilde a_\psi(n,K)
& =
\frac{2}{K} \sum_{k \ge 0} e^{-2k/K} a_\psi(n,k) \\
& =
\frac{2}{K} e^{2/K} \int_0^{2\pi} \left| \left\langle \delta_n , \left( \E - e^{1/K +i\theta} \right)^{-1} \psi \right\rangle \right|^2 \, \frac{d\theta}{2\pi}.
\end{align*}
To complete the proof, denote $B_K' = \{ \theta \in [0,2\pi) : e^{i\theta} \in B_K \}$, and sum over $n \geq R_K/3$ to get
\begin{align*}
\sum_{n \ge R_K / 3} \widetilde a_\psi(n,K)
& \ge
\frac{2}{K} e^{2/K} \int_{B_K'} \sum_{n \ge R_K/3}  \left| \left\langle \delta_n , \left( \E - e^{1/K +i\theta} \right)^{-1} \psi \right\rangle \right|^2 \, \frac{d\theta}{2\pi} \\
& \ge
\frac{\widetilde C}{K} |B_K| R_K^{1-2\gamma} \\
& =
\widetilde C |B_K| R_K^{-3 \gamma},
\end{align*}
which proves \eqref{eq:poutbound}. In the calculation above, we have used \eqref{eq:pout:lb} to obtain the second line and $K = R_K^{1+\gamma}$ to get the third. Applying \eqref{eq:moments:pout} to \eqref{eq:poutbound} yields \eqref{eq:momentbound}.
\end{proof}

\section{Upper Bounds on Spreading}

In this section, we will prove the bounds on left and right time-averaged probabilities from Theorem~\ref{DT07.thm7}. The key ingredient is again Lemma~\ref{l.parseval}, which gives us the following relation between each time-averaged single-site probability and a corresponding average of a matrix element of the resolvent off of $\partial \D$:
\begin{equation} \label{eq:parseval}
\widetilde a_{\delta_{-1}}(n,K)
=
\frac{2}{K} e^{2/K} \int_0^{2\pi} \left| \left\langle \delta_n , \left( \E - e^{1/K +i\theta} \right)^{-1} \delta_{-1} \right\rangle \right|^2 \, \frac{d\theta}{2\pi}.
\end{equation}
Let's take $\varepsilon = 1/K$, and define $G(z) = (\E - z)^{-1}$ for $z \in \C \setminus \sigma(\E)$. Using \eqref{eq:parseval} and the inequality $e^2 < 3\pi$, we then have
\begin{align}
\label{Pr}
\widetilde P_r(N,K)
& <
3 \varepsilon \int_{0}^{2\pi} \! M_r\left(N, e^{\varepsilon + i\theta}\right) \, d\theta, \\
\label{Pl}
\widetilde P_l(N,K)
& <
3 \varepsilon \int_{0}^{2\pi} \! M_l\left(N, e^{\varepsilon + i\theta}\right) \, d\theta,
\end{align}
where we have defined
\begin{align}
\label{Mr}
M_r(N,z)
=
\sum_{n > N} \vert \left\langle \delta_n, G(z)\delta_{-1} \right\rangle\vert^2, \\
\label{Ml}
M_l(N,z)
=
\sum_{n < -N} \vert \left\langle \delta_n, G(z)\delta_{-1} \right\rangle\vert^2.
\end{align}
By symmetry, we will only prove the bounds on $\widetilde P_{r}$, since the argument which proves the bounds on $\widetilde P_{l}$ is nearly identical. The overall strategy is to approximate $\E$ by suitable sequences of ``truncations''  which converge strongly to $\E$. One can then prove appropriate bounds on the quantities  in \eqref{Mr} corresponding to the truncations (Lemmas~\ref{l:greendecay1} and \ref{l:greendecay2}). Next, one controls the error incurred in passing between $\E$ and the truncations (Lemma~\ref{l:MSbounds}). Finally, one puts all of this together to bound $M_r$ (and hence $\widetilde P_r$) in terms of the GZ cocycle (Lemma~\ref{DT07.Lemma3}). The remainder of the section makes this discussion precise.

\bigskip

Given $N \in \Z_+$, define four ``truncated'' CMV matrices $\mathcal E_{2N}$, $\widetilde{\E}_{2N}$, $\E_{4N}'$, and $\widetilde{\E}_{4N}'$ by choosing Verblunsky coefficients as follows:
\begin{align*}
\alpha_{2N}(n)
& =
\begin{cases} \alpha(n) &\mbox{when } n < 2N,\\
0 &\mbox{when }{n\geq 2N} \mbox{ and } n \mbox{ even},\\
-\frac{3}{4} &\mbox{when }{n> 2N} \mbox{ and } n \mbox{ odd},
\end{cases} \\
\widetilde{\alpha}_{2N}(n)
& =
\begin{cases} \alpha(n) &\mbox{when } n < 2N,\\
0 &\mbox{when }{n\geq 2N} \mbox{ and } n \mbox{ even},\\
\frac{3}{4} &\mbox{when }{n> 2N} \mbox{ and } n \mbox{ odd},
\end{cases} \\
\alpha_{4N}'(n)
& =
\begin{cases} \alpha(n) &\mbox{when } n < 4N,\\
0 &\mbox{when }{n\geq 4N} \mbox{ and } n \mbox{ even},\\
(-1)^{(n+1)/2} \frac{3}{4} &\mbox{when }{n> 4N} \mbox{ and } n \mbox{ odd},
\end{cases}\\
\widetilde{\alpha}_{4N}'(n)
& =
\begin{cases} \alpha(n) &\mbox{when } n < 4N,\\
0 &\mbox{when }{n\geq 4N} \mbox{ and } n \mbox{ even},\\
(-1)^{(n-1)/2}\frac{3}{4} &\mbox{when }{n> 4N} \mbox{ and } n \mbox{ odd}.
\end{cases}
\end{align*}
For convenience, let us also define
\begin{align*}
\varepsilon_0
& =
\log\left(\frac{4}{\sqrt 7} \right), \\
\mathcal R
& =
\left\{ z = e^{i\theta + \varepsilon} \in \C :
\theta \in [-\pi/4, \pi/4] \cup [3\pi/4,5\pi/4], \; 0 < \varepsilon < \varepsilon_0 \right\}, \\
\mathcal R'
& =
\left\{ z = e^{i\theta + \varepsilon} \in \C :
\theta \in [\pi/4, 3\pi/4] \cup [5\pi/4,7\pi/4], \; 0 < \varepsilon < \varepsilon_0 \right\}.
\end{align*}
Let $(p_+(z, n),r_+(z, n))^\top$ be a solution to the GZ recursion subject to the initial conditions $p_+(z, -1) = z, \, r_+(z, -1) = 1$. Similarly, let $(q_+(z, n), s_+(z, n))^\top$ be a solution to the GZ recursion with $q_+(z, -1) = z, \, s_+(z, -1) = -1$. More precisely,
\begin{equation}\label{pqrs.definition}
\begin{pmatrix}
p_+(z, n) & q_+(z, n) \\
r_+(z, n) & s_+(z, n)
\end{pmatrix}
=
Z(n+1,0;z)
\begin{pmatrix}
z & z \\ 1 & -1
\end{pmatrix},
\end{equation}
where $Z$ is as in \eqref{gz:def}. Since $\det(P(\alpha,z)) = \det(Q(\alpha,z)) = -1$, we easily compute the Wronskian of these solutions:
\begin{equation} \label{eq:gzwronsk}
W(z,n)
:=
p_+(z,n) s_+(z,n) - q_+(z,n) r_+(z,n)
=
2(-1)^n z.
\end{equation}
In particular, $|W| = 2|z|$ does not depend on $n$; moreover $W$ is uniformly bounded on $\mathcal R \cup \mathcal R'$.

Let $M_\pm(z)$ denote the half-lattice Weyl-Titchmarsh $m$-functions corresponding to our operator $\mathcal E$. Readers unfamiliar with the connection between these functions and the spectral theory of $\mathcal E$ should consult Appendix A of \cite{GZ06}.

The functions $M_\pm$ are defined precisely in \cite[(2.137)--(2.140)]{GZ06}, but here we will merely summarize the properties of $M_\pm$ most relevant to our proof. $M_+$ is a Carath\'eodory function, that is, an analytic function from $\mathbb D$ to the open right half plane. $M_-$ is an anti-Carath\'eodory function, that is, an analytic function from $\mathbb D$ to the open left half plane. One of the reasons that these functions are important with regard to understanding the spectral theory of $\mathcal E$ is that for $z\in\partial \D$ we have this useful Herglotz representation,
\begin{equation} \label{eq:mdef}
M_+(z)
=
\int_{\partial\D} \! \frac{\tau + z}{\tau - z} \, d\mu_+(\tau),
\end{equation}
where $\mu_+$ denotes the spectral measures of the restriction of $\E$ to the right half-line (see \cite[(2.123) and (2.137)]{GZ06}). There is also a similar formula for $M_-(z)$, but it is not relevant for our purposes. Using \cite[(2.57), (2.61), (2.62), (3.5), and (3.7)]{GZ06}, (splitting the CMV operator the way we do means setting $k_0 = -1$) we have
\begin{equation}\label{GZ.formula}
\left\langle (\mathcal E - z)^{-1}\delta_{-1}, \delta_n \right\rangle
=
d(z)s_+(z, n) + b(z)r_+(z, n),
\end{equation}
for $n > -1$, where
\begin{equation} \label{eq:bd:def}
d(z)
=
-\frac{1 + M_-(z)}{2z(M_+(z) - M_-(z))},
\quad
b(z)
=
-\frac{(1 + M_-(z)) M_+(z)}{2 z (M_+(z) - M_-(z))}.
\end{equation}
Let $\varepsilon > 0$ and $z \in \C$ with $\log|z| = \varepsilon$ be given. From \eqref{pqrs.definition} and \eqref{GZ.formula}, we obtain
\begin{equation} \label{eq:bddiffbound}
|b(z) - d(z)|
=
|\langle \delta_{-1}, (\E-z)^{-1} \delta_{-1} \rangle|
\le
\frac{1}{e^\varepsilon - 1}
\le
\varepsilon^{-1}.
\end{equation}
The penultimate inequality is a standard consequence of the spectral theorem, while the final inequality follows from $\varepsilon > 0$. Notice that $\mathrm{Re}(M_+(z)) < 0$. Indeed, it is easy to see that
$$
\mathrm{Re}\left( \frac{\tau + z}{\tau - z} \right)
=
\frac{1 - |z|^2}{|\tau - z|^2}
<
0
$$
for every $\tau$ with $|\tau| = 1$, so, by \eqref{eq:mdef}, we have $\mathrm{Re}(M_+(z)) < 0$. Consequently,
\begin{equation} \label{eq:bdargbound}
\mathrm{arg}\left(\frac{b(z)}{d(z)} \right)
=
\mathrm{arg}(M_+(z))
\in
\left( \frac{\pi}{2}, \frac{3\pi}{2} \right).
\end{equation}
Combining \eqref{eq:bddiffbound} and \eqref{eq:bdargbound}, we obtain
\begin{equation}\label{bd.bound}
\vert b(z) \vert \le \varepsilon^{-1}
\text{ and }
\vert d(z) \vert \le \varepsilon^{-1}.
\end{equation}
Notice that the bound in \eqref{bd.bound} is a purely deterministic result. It holds for any choice of $\alpha$ and then for any $z \in \C$ with $\log|z| = \varepsilon$.

\bigskip

Let us consider the GZ transfer matrices $Y_{2N}(n, z)$ corresponding to the truncated Verblunsky coefficients $\alpha_{2N}$. Note that for $n \geq N$, we have
$$
Y_{2N}(2n+1, z) Y_{2N}(2n, z)
=
\frac{1}{\sqrt{1 - \left( \frac{3}{4} \right)^2}}
\begin{pmatrix}
z & \frac{3}{4}\\
\frac{3}{4} & z^{-1}
\end{pmatrix}.
$$
This matrix has determinant $1$. We can calculate that the eigenvalues are
\begin{equation}\label{lambda.eigenvalues}
\lambda_\pm
=
\lambda_\pm(z)
:=
\frac{2(z^2 + 1) \pm \sqrt{4z^4 + z^2 + 4}}{\sqrt{7} z}.
\end{equation}
We will wait until after Lemma~\ref{l:evnotoncircle} to choose a branch of the square root  $\sqrt{4z^4 + z^2 + 4}$. Notice that $\lambda_\pm$ in our notation means something different than $\lambda_\pm$ in \cite{DT07}. After a short calculation, we see that the eigenvectors $e^\pm$ of $Y_{2N}(2n + 1, z) Y_{2N}(2n, z)$ are given by
\begin{equation}\label{lambda.eigenvectors}
e^\pm
=
e^\pm(z)
=
\begin{pmatrix}
e^\pm_1(z)\\
e^\pm_2(z)
\end{pmatrix}
=
\frac{1}{3}
\begin{pmatrix}
\sqrt{7} \lambda_\pm(z) - 4z^{-1} \\
3
\end{pmatrix}
\end{equation}
for all $n \geq N$.

\begin{lemma} \label{l:evnotoncircle}
For all $z \in \mathcal R$, one of $\lambda_\pm(z)$ has modulus strictly greater than $1$, and the other has modulus strictly less than $1$. In fact, $\lambda_-$ is bounded away from $\partial \D$ on $\mathcal R$, i.e., there exists a constant $c_0 < 1$ with $|\lambda_-(z)| < c_0$ for all $z \in \mathcal R$.
\end{lemma}

\begin{proof}
Consider $\E_{-\infty}$, i.e., the 2-periodic extended CMV matrix with alternating Verblunsky coefficients $\ldots, 0, -3/4, 0, -3/4, \ldots$. Using standard calculations from Floquet theory, one can verify that
$$
\Sigma
:=
\sigma(\E_{-\infty})
=
\left\{ z \in \partial \D : |\mathrm{Im}(z)| \ge \frac{3}{4} \right\}.
$$
On the other hand, it is easy to see that $z$ is a generalized eigenvalue of $\E_{-\infty}$ if and only if $|\lambda_+(z)| = |\lambda_-(z)| = 1$ (see \cite[Section~3]{DFLY2} for definitions and details). Thus, since $\mathcal R \cap \Sigma = \emptyset$, the first statement of the lemma follows. To verify the second statement of the lemma, simply note that $\lambda_+$ and $\lambda_-$ are continuous functions of the spectral parameter away from $z = 0$ and that $\mathcal R$ is bounded away from $\Sigma$.
\end{proof}

For $N \in \Z_+$ and $|z| > 1$, define $G_{2N}(z) = (\E_{2N} - z)^{-1}$.

\begin{lemma} \label{l:greendecay1}
For $n\geq 2N\geq 4$ and $z = e^{i\theta + \varepsilon} \in \mathcal R$,
\begin{align}
\label{eq:greendecay1}
\left\lvert \left\langle \delta_n, G_{2N}(z) \delta_{-1} \right\rangle \right\rvert
& \lesssim
\varepsilon^{-1}
\frac{|z| \vert \lambda_-(z) \vert^{\frac{n}{2} - N}}{\left\vert e_1^-(z) p_+(z, 2N-1)-r_+(z, 2N-1) \right\vert},\\
\label{eq:greendecay2}
\left\lvert \left\langle \delta_n,  G_{2N}(z) \delta_{-1} \right\rangle \right\rvert
& \lesssim
\varepsilon^{-1}
\frac{|z| \vert \lambda_-(z)  \vert^{\frac{n}{2} - N}}{\left\vert e_1^-(z) q_+(z, 2N-1)-s_+(z, 2N-1) \right\vert}.
\end{align}
\end{lemma}

\begin{proof}
Fix $N \in \Z_+$ and consider a GZ recurrence corresponding to $\alpha_{2N}$.  Let $u = G_{2N}(z) \delta_{-1} =  (\mathcal E_{2N} - z)^{-1} \delta_{-1}$; then $u \in \ell^2(\Z)$ and $u$ satisfies $\mathcal E_{2N} u = zu + \delta_{-1}$; in particular, $u$ satisfies $\E_{2N} u = zu$ away from site $-1$. Let $v = \mathcal L^{-1} u$ and $\Phi = (u,v)^\top$, as usual (notice that $v \in \ell^2$ also). We can write
\begin{equation} \label{eq:cutoffgz:ic}
\Phi(2N - 1)
=
C_+ e^+ + C_- e^-,
\end{equation}
for some coefficients $C_\pm = C_\pm(z, 2N - 1) \in \C$, where $e^\pm$ are defined in \eqref{lambda.eigenvectors}.   Using \eqref{eq:gzsoltransfer}, and $2N \ge 4$, we have
\begin{equation} \label{eq:cutoffgz:genterm}
\Phi(2m - 1)
=
C_+ \lambda_+^{m-N} e^+ + C_- \lambda_-^{m-N} e^-
\text{ for all } m \geq N.
\end{equation}
Note that \eqref{eq:cutoffgz:genterm} doesn't hold for $\Phi(2m)$, since $e^\pm$ are the eigenvalues of the two-step matrices $Y_{2N}(2n+1,z) Y_{2N}(2n,z)$. However, using
$$
Y_{2N}(2m,z)
=
\begin{pmatrix} 0 & 1 \\ 1 & 0
\end{pmatrix}
\text{ for } m \ge N,
$$
we must have $u(2m) = v(2m-1)$ and $v(2m) = u(2m-1)$ whenever $m \ge N$.

Consequently, since the entries of $\Phi$ are in $\ell^2$ we must have $C_+ = 0$. We thus have
$$
\left\langle \delta_{2N-1}, G_{2N}(z) \delta_{-1} \right\rangle
=
u(2N - 1)
=
C_-(z, 2N - 1) e_1^-(z),
$$
and
$$
\left\langle \delta_{2N}, G_{2N}(z) \delta_{-1} \right\rangle
=
u(2N)
=
v(2N - 1)
=
C_-(z, 2N - 1) e_2^-(z).
$$
Using \eqref{GZ.formula}, we have
\begin{align}
d_{2N}(z) s_+(z, 2N - 1) + b_{2N}(z) r_+(z, 2N-1)
& =
C_-(z,2N-1) e_1^-(z),
\label{Resolvent.eq1}\\
d_{2N}(z) q_+(z,2N-1) + b_{2N}(z) p_+(z,2N-1)
& =
C_-(z,2N-1) e^-_2(z).
\label{Resolvent.eq2.mod}
\end{align}
where $b_{2N}$ and $d_{2N}$ are defined by \eqref{eq:bd:def} with $\E$ replaced by $\E_{2N}$. More specifically, notice that $Z(n+1,0;z)$ only depends on $\alpha_0, \ldots, \alpha_n$. In particular, the $p_+,q_+,r_+,s_+$-functions associated to $\E$ and $\E_{2N}$ agree for $-1 \le n \le 2N-1$, which suffices to prove \eqref{Resolvent.eq1} by applying \eqref{GZ.formula}. For \eqref{Resolvent.eq2.mod}, notice that $\alpha_{2N}(2N) = 0$ implies
\begin{align*}
s_+^{2N}(z,2N)
& =
q_+^{2N}(z,2N-1)
=
q_+(z,2N-1) \\
r_+^{2N}(z,2N)
& =
p_+^{2N}(z,2N-1)
=
p_+(z,2N-1),
\end{align*}
where the superscript $2N$'s denote quantities associated to $\E_{2N}$. Recall \eqref{eq:gzwronsk}: the Wronskian is given by
$$
W(z,n)
:=
p_+(z,n) s_+(z,n) - q_+(z,n) r_+(z,n)
=
2 (-1)^n z.
$$
Solving \eqref{Resolvent.eq1} and \eqref{Resolvent.eq2.mod} for $C_-$, we get
\begin{align}
\label{Cminus-pr}
C_-
& =
\frac{Wd_{2N}}{e_1^- p_+(2N-1) - e_2^- r_+(2N-1)} \\
\label{Cminus-qs}
C_-
& =
\frac{Wb_{2N}}{- e_1^- q_+(2N-1) + e_2^- s_+(2N-1)}.
\end{align}
In light of \eqref{eq:cutoffgz:genterm}, this yields
\begin{align*}
\left\langle \delta_{2m-1}, G_{2N}(z) \delta_{-1} \right\rangle
& =
u(2m - 1)
=
\frac{Wd_{2N} \lambda_-^{m-N} e_1^-}{e_1^- p_+(2N-1) - e_2^- r_+(2N-1)} \\
\left\langle \delta_{2m-1}, G_{2N}(z) \delta_{-1} \right\rangle
& =
u(2m - 1)
=
\frac{Wb_{2N} \lambda_-^{m-N} e_1^-}{-e_1^- q_+(2N-1) + e_2^- s_+(2N-1)}
\end{align*}
The statement of the lemma for $n \geq 2N$ odd then follows immediately from \eqref{eq:gzwronsk},  \eqref{bd.bound}, and \eqref{lambda.eigenvectors}. The statement for even $n \geq 2N$ follows from
$$
\langle \delta_{2m}, G_{2N}(z) \delta_{-1} \rangle
=
u(2m)
=
v(2m-1)
\text{ for all } m \geq N.
$$

\end{proof}

Using a completely analogous calculation, we obtain the following lemma for the Green's function of $\widetilde{\E}$.

\begin{lemma} \label{l:greendecay2}
 Define $\widetilde G_{2N}(z) = \left(\widetilde \E_{2N} - z \right)^{-1}$ for $|z| > 1$. Then, for $z = e^{i\theta + \varepsilon} \in \mathcal R$ and $n \geq 2N \geq 4$,
\begin{align}
\left\lvert \left\langle \delta_n,  \widetilde G_{2N}(z) \delta_{-1} \right\rangle \right\rvert
\lesssim
\varepsilon^{-1} \frac{\vert z\vert \vert \lambda_-(z) \vert^{\frac{n}{2} - N}}{\left\vert e_1^-(z) p_+(2N-1,z) + r_+(2N-1,z) \right\vert}\\
\left\lvert \left\langle \delta_n, \widetilde G_{2N}(z)\delta_{-1} \right\rangle \right\rvert
\lesssim
\varepsilon^{-1} \frac{\vert z\vert \vert \lambda_-(z)\vert^{\frac{n}{2} - N}}{\left\vert e_1^-(z) q_+(2N-1,z) + s_+(2N-1,z) \right\vert}
\end{align}
\end{lemma}

Given $N \in \Z_+$, let us denote $\underline{N} = 2\lfloor N/2\rfloor$. For $z \notin \partial \D$, define
\begin{align*}
S(N,z)
& =
\sum_{n = N + 1}^\infty
\Big\lvert \Big\langle \delta_n, G_{\underline{N}}(z) \delta_{-1} \Big\rangle \Big\rvert^{2},\\
\widetilde S(N,z)
& =
\sum_{n=N+1}^\infty
\left\lvert \left\langle \delta_n, \widetilde G_{\underline{N}}(z) \delta_{-1} \right> \right\rvert^2.
\end{align*}

\begin{lemma} \label{l:MSbounds}
For all $z = e^{i\theta + \varepsilon} \in \mathcal R$, we have
\begin{align}
\varepsilon^2S(N,z)\lesssim M_r(N,z)
& \lesssim
\varepsilon^{-2}S(N, z),\\
\varepsilon^2\widetilde S(N,z)\lesssim M_r(N,z)
& \lesssim
\varepsilon^{-2}\widetilde S (N, z).
\end{align}
\end{lemma}
\begin{proof}
Using a standard resolvent identity (e.g.~\cite[Lemma~6.5]{Teschl}), we get
$$
(\E - z)^{-1} - (\E_{\underline{N}} - z)^{-1}
=
-(\E - z)^{-1}\left(\E - \E_{\underline{N}}\right)
(\E_{\underline{N}} - z)^{-1}.
$$
Since $\vert\vert (\E - z)^{-1}\vert\vert\leq \varepsilon^{-1}$ we then have
\begin{align*}
M_r(N, z)
& =
\sum_{n = N+1}^\infty | \langle \delta_n, G(z) \delta_{-1} \rangle |^2 \\
& =
\sum_{n = N+1}^\infty | \langle \delta_n, G_{\underline{N}}(z)\delta_{-1}
-G(z)\left(\E - \E_{\underline{N}}\right) G_{\underline{N}}(z) \delta_{-1} \rangle |^2 \\
& \leq
2 S(N,z) + 2\varepsilon^{-2}\lVert(\E - \E_{\underline{N}})\rVert^2
\sum_{n = N + 1}^\infty
\left\vert \left\langle \delta_n, G_{\underline{N}}(z)\delta_{-1} \right\rangle \right\vert^2\\
& \lesssim
\varepsilon^{-2}S(N,z).
\end{align*}
In the first inequality, we have used $|a+b|^2 \leq 2|a|^2 + 2|b|^2$ for complex numbers $a$ and $b$, which follows readily from Cauchy--Schwarz. The proof is similar for the other inequality. The proof in the $\widetilde S$ case is identical.
\end{proof}

\begin{lemma}\label{DT07.Lemma3}
For all $z = e^{i\theta + \varepsilon} \in \mathcal R$, we have
\begin{equation} \label{eq:greensum:transmatbound}
M_r(N,z)
\lesssim
\varepsilon^{-4}\left( \max_{4\leq n\leq N} \left\|  Z(n, 0; z) \right\|^2 \right)^{-1}.
\end{equation}
The implicit constants depend on $M = \|\alpha\|_\infty$, but nothing else.
\end{lemma}

\begin{proof}
The previous lemmas show us that
\begin{align}
\label{eq:mrboundpr}
M_r(N,z)
\lesssim
\varepsilon^{-4} \left\vert \left( \frac{\sqrt 7}{3} \lambda_-(z) - \frac{4}{3} z^{-1} \right) p_+(\underline{N}-1,z) - r_+(\underline{N}-1,z) \right\vert^{-2},\\
\label{eq:mrboundpr2}
M_r(N,z)
\lesssim
\varepsilon^{-4} \left\vert \left( \frac{\sqrt 7}{3} \lambda_-(z) - \frac{4}{3} z^{-1} \right) p_+(\underline{N}-1,z) + r_+(\underline{N}-1,z) \right\vert^{-2},\\
\label{eq:mrboundqs}
M_r(N,z)
\lesssim
\varepsilon^{-4} \left\vert \left( \frac{\sqrt 7}{3} \lambda_-(z) - \frac{4}{3} z^{-1} \right) q_+(\underline{N}-1,z) - s_+(\underline{N}-1,z) \right\vert^{-2},\\
\label{eq:mrboundqs2}
M_r(N,z)
\lesssim
\varepsilon^{-4} \left\vert \left( \frac{\sqrt 7}{3} \lambda_-(z) - \frac{4}{3} z^{-1} \right) q_+(\underline{N}-1,z) + s_+(\underline{N}-1,z) \right\vert^{-2}.
\end{align}
Notice that \eqref{eq:mrboundpr} and \eqref{eq:mrboundqs} implicitly use the claim from Lemma~\ref{l:evnotoncircle} that $|\lambda_-| < c_0$ on $\mathcal R$ to uniformly bound
\begin{equation} \label{eq:lambdageosumbound}
\sum_{j=0}^\infty |\lambda_-(z)|^j
\leq
\frac{1}{1-c_0}
<
\infty
\end{equation}
for all $z \in \mathcal R$. This step explains the need for several truncations, since we cannot have a bound like \eqref{eq:lambdageosumbound} on $\mathcal R'$, which means that we lose control of the implicit constants in that region with this pair of truncations. Let us denote
\begin{align}
U
=
\left( \frac{\sqrt 7}{3} \lambda_-(z) - \frac{4}{3} z^{-1} \right) p_+(\underline{N}-1,z)-r_+(\underline{N}-1,z),\\
\widetilde U
=
\left(\frac{4}{3} z^{-1} - \frac{\sqrt 7}{3} \lambda_-(z) \right) p_+(\underline{N}-1,z)-r_+(\underline{N}-1,z).
\end{align}
Adding and subtracting the two equations we note that
$$
\vert U\vert+\vert\widetilde U\vert
\geq
\frac{2}{3} \left| \sqrt 7 \lambda_-(z) - 4z^{-1} \right| \vert p_+(\underline{N}-1,z)\vert,
$$
and also
$$
\vert U\vert +\vert\widetilde U\vert
\geq
2\vert r_+(\underline{N}-1,z)\vert.
$$
Since $\vert \lambda_-\vert < 1$ and $1 < \vert z\vert < 4 / \sqrt 7$, we must have $\sqrt 7 \lambda_-(z)- 4z^{-1} \neq 0$.

This implies

\begin{equation}\label{Mbound.pr}
M_r(N,z)
\lesssim \varepsilon^{-4}\frac{1}{(\lvert U\rvert+\lvert\widetilde U\rvert)^2}\lesssim\varepsilon^{-4} \left( \vert p_+(\underline{N}-1,z)\vert^2+\vert r_+(\underline{N}-1,z)\vert^2\right)^{-1}.
\end{equation}
We can similarly also show that
\begin{equation}\label{Mbound.qs}
M_r(N,z)
\lesssim
\varepsilon^{-4} \left( \vert q_+(\underline{N}-1,z)\vert^2+\vert s_+(\underline{N}-1,z)\vert^2 \right)^{-1}.
\end{equation}
For the sake of notational convenience, in the following calculation let us use $p_+,q_+,r_+,s_+$ as a shorthand for $p_+(\underline{N}-1,z), q_+(\underline{N}-1,z), r_+(\underline{N}-1,z), s_+(\underline{N}-1,z)$.
We can use \eqref{pqrs.definition} to observe that
\begin{align*}
\lVert Z(\underline{N},0;z)\rVert^2
& =
\sup_{(A,B)\neq(0,0)}\frac{\left\lVert A\begin{pmatrix}
p_+\\r_+
\end{pmatrix}+B\begin{pmatrix}
q_+\\s_+
\end{pmatrix}\right\rVert^2}{\left\lVert A\begin{pmatrix}
z\\1
\end{pmatrix}+B\begin{pmatrix}
z\\-1
\end{pmatrix}\right\rVert^2}\\
& \leq
\sup_{(A,B)\neq(0,0)}\frac{\left( \lvert A\rvert^2+\lvert B\rvert^2 \right) \left( \lvert p_+\rvert^2+\lvert q_+\rvert^2+\lvert r_+\rvert^2+\lvert s_+\rvert^2 \right)}
{\lvert A + B \rvert^2+\lvert A-B\rvert^2}\\
& =
\frac{1}{2} \left( \lvert p_+\rvert^2+\lvert q_+\rvert^2+\lvert r_+\rvert^2+\lvert s_+\rvert^2 \right)\\
& \le
\max(\lvert p_+\rvert^2+\lvert r_+\rvert^2, \lvert q_+\rvert^2+\lvert s_+\rvert^2).
\end{align*}
Since $\alpha$ is bounded away from $\partial \D$, we have $Z(N,0;z) \lesssim Z(\underline{N},0;z)$, which implies
$$
M_r(N,z)
\lesssim
\varepsilon^{-4}\lVert Z(N,0;z)\rVert^{-2}.
$$
Observe that $M_r(n,z) \geq M_r(N,z)$ whenever $n \le N$. Thus, the previous argument gives the following for $4 \le n \le N$:
$$
M_r(N,z)
\leq
M_r(n,z)
\lesssim
\varepsilon^{-4} \left( \| Z(n,0;z) \|^{2} \right)^{-1}.
$$
Since the preceding inequality holds uniformly in $4 \le n \le N$, we get \eqref{eq:greensum:transmatbound}.
\end{proof}

We  want to extend all these results to arbitrary $\theta\in [0,2\pi)$ by using the second pair of truncations of $\E$. We can then obtain a new version of Lemmas~\ref{l:greendecay1} and \ref{l:greendecay2}. Notice that the four-step GZ matrices of the truncated CMV operators take the form
\begin{align*}
Z_{4N}'(4n+4,4n;z)
=
\frac{1}{7} \begin{pmatrix}
16z^2 - 9 & 12(z-z^{-1}) \\
-12(z-z^{-1}) & 16z^{-2} - 9
\end{pmatrix} \\
\widetilde Z_{4N}'(4n+4,4n;z)
=
\frac{1}{7} \begin{pmatrix}
16z^2 - 9 & -12(z-z^{-1}) \\
12(z-z^{-1}) & 16z^{-2} - 9
\end{pmatrix}
\end{align*}
for $n \ge N$. Both have eigenvalues
$$
\lambda_\pm'(z)
=
\frac{8(z^2 + z^{-2}) - 9 \pm \sqrt{(8(z^2 + z^{-2}) - 9)^2 - 49 } }{7}
$$
The first has eigenvectors $v_\pm = ( v_1^\pm, v_2^\pm)^\top$,
given by
$$
v_1^\pm
=
\frac{1}{12(z-z^{-1})}\left(7\lambda_\pm' - 8(z^2+z^{-2})  + 9 \right) - \frac{2}{3}(z+z^{-1}),
\quad
v_2^\pm =1.
$$
The second has eigenvector $(-v_1^\pm, 1)^\top$. Following previous arguments, we can prove bounds for the decay of the corresponding Green's functions as before.

\begin{lemma} \label{l':greendecay1}
Define $G'_{4N}(z) = \left(\E'_{4N} - z \right)^{-1}$ for $|z| > 1$. For $n\geq 4N\geq 8$ and $z = e^{i\theta + \varepsilon} \in \mathcal R'$,
\begin{align}
\label{eq':greendecay1}
\left\lvert \left\langle \delta_n, G'_{4N}(z) \delta_{-1} \right\rangle \right\rvert
& \lesssim
\varepsilon^{-1}
\frac{|z| \vert \lambda'_-(z) \vert^{\frac{n}{4} - N}}{\left\vert v_1^-(z) p_+(z, 4N-1)-r_+(z, 4N-1) \right\vert},\\
\label{eq':greendecay2}
\left\lvert \left\langle \delta_n,  G'_{4N}(z) \delta_{-1} \right\rangle \right\rvert
& \lesssim
\varepsilon^{-1}
\frac{|z| \vert \lambda'_-(z)  \vert^{\frac{n}{4} - N}}{\left\vert v_1^-(z) q_+(z, 4N-1)-s_+(z, 4N-1) \right\vert}.
\end{align}
\end{lemma}

\begin{lemma} \label{l':greendecay2}
Define $\widetilde G'_{4N}(z) = \left(\widetilde \E'_{4N} - z \right)^{-1}$ for $|z| > 1$. For $n\geq 4N\geq 8$ and $z = e^{i\theta + \varepsilon} \in \mathcal R'$,
\begin{align}
\left\lvert \left\langle \delta_n,  \widetilde G'_{4N}(z) \delta_{-1} \right\rangle \right\rvert
\lesssim
\varepsilon^{-1} \frac{\vert z\vert \vert \lambda'_-(z) \vert^{\frac{n}{4} - N}}{\left\vert v_1^-(z) p_+(4N-1,z) + r_+(4N-1,z) \right\vert}\\
\left\lvert \left\langle \delta_n, \widetilde G'_{4N}(z)\delta_{-1} \right\rangle \right\rvert
\lesssim
\varepsilon^{-1} \frac{\vert z\vert \vert \lambda'_-(z)\vert^{\frac{n}{4} - N}}{\left\vert v_1^-(z) q_+(4N-1,z) + s_+(4N-1,z) \right\vert}
\end{align}
\end{lemma}

For $z \notin \partial \D$, define
\begin{align*}
S'(N,z)
& =
\sum_{n = N + 1}^\infty
\Big\lvert \Big\langle \delta_n, G'_{4\lfloor N/4 \rfloor}(z) \delta_{-1} \Big\rangle \Big\rvert^{2},\\
\widetilde S'(N,z)
& =
\sum_{n=N+1}^\infty
\left\lvert \left\langle \delta_n, \widetilde G'_{4\lfloor N/4 \rfloor}(z) \delta_{-1} \right> \right\rvert^2.
\end{align*}
A new version of Lemma~\ref{l:MSbounds} is also immediate.

\begin{lemma} \label{l':MSbounds}
For $z = e^{i\theta + \varepsilon} \in \mathcal R'$, we have
\begin{align}
\varepsilon^2S'(N,z)\lesssim M_r(N,z)
& \lesssim
\varepsilon^{-2}S'(N, z),\\
\varepsilon^2\widetilde S'(N,z)\lesssim M_r(N,z)
& \lesssim
\varepsilon^{-2}\widetilde S' (N, z).
\end{align}
\end{lemma}

Notice that $v_1^-(z)$ is bounded away from zero for $z \in \mathcal R'$. To see this, notice that $(0,1)^\top$ is an eigenvector of $Z'_{4N}(4N+4,4N;z)$ if and only if $z-z^{-1} = 0$, that is $z = \pm 1$. Since $\mathcal R'$ is bounded away from $\{ \pm 1 \}$, it follows that $v_1^-(z)$ is bounded away from zero. This observation allows us to extend the arguments (and hence the conclusions) of Lemma~\ref{DT07.Lemma3} to spectral parameters in $\mathcal R'$.

\begin{lemma}\label{DT07.Lemma3'}
For $z = e^{i\theta + \varepsilon} \in \mathcal R'$, we have
\begin{equation} \label{eq':greensum:transmatbound}
M_r(N,z)
\lesssim
\varepsilon^{-4}\left( \max_{4\leq n\leq N} \left\|  Z(n, 0; z) \right\|^2 \right)^{-1}.
\end{equation}
The implicit constants depend on $M = \|\alpha\|_\infty$, but nothing else.
\end{lemma}

\begin{proof}[Proof of Theorem~\ref{DT07.thm7}]

Follows immediately from Lemmas \ref{DT07.Lemma3} and \ref{DT07.Lemma3'}, except that we have the $\max$ in \eqref{eq:pr:gzbound} running over $n\geq 1$ instead of $n\geq 4$. Since $\alpha$ is bounded away from $\partial \D$, making the appropriate adjustments only modifies the implicit constants in an $N$-independent fashion, so this is fine.

\end{proof}

\section{Upper Bounds Without Averaging}

By using techniques from complex analysis and a version of the residue calculus for normal operators, we can prove variants of the upper bounds from the previous section without averaging in time; compare \cite{DT08}.

\begin{lemma} \label{l:contourint}
Let $\E$ be an extended CMV matrix. We have
\begin{equation} \label{eq:contourint}
\left\langle \delta_n, \E^k \delta_{-1} \right\rangle
=
-\frac{1}{2\pi i}
\int_{\Gamma} \! z^k \left\langle \delta_n, (\E - z)^{-1} \delta_{-1} \right\rangle \, dz
\end{equation}
for every $k \in \Z$, all $n \in \Z$, and every positively oriented contour $\Gamma$ in $\C$ which encloses $\overline{\D}$.
\end{lemma}

\begin{proof}
This follows from the Dunford functional calculus for normal operators. See \cite{Dun43, DunfordSchwartz1, ReedSimon1}.
\end{proof}

\begin{lemma} \label{l:outbound:no:avg}
Let $\E$ be an extended CMV matrix, and consider the unaveraged left and right probabilities $P_r$ and $P_l$. We have
\begin{align}
\label{eq:rightbd:no:avg}
P_r(N,k)
& \lesssim
\int_0^{2 \pi}  \sum_{n > N} \left| \left\langle \delta_n,
\left( \E - e^{1/k + i \theta} \right)^{-1} \delta_{-1}
\right\rangle \right|^2 \frac{d\theta}{2\pi} \\
\label{eq:leftbd:no:avg}
P_l(N,k)
& \lesssim
\int_0^{2 \pi}  \sum_{n < -N} \left| \left\langle \delta_n,
\left( \E - e^{1/k + i \theta} \right)^{-1} \delta_{-1}
\right\rangle \right|^2 \frac{d\theta}{2\pi}
\end{align}
for all $k \in \Z \setminus \{ 0 \}$.
\end{lemma}

\begin{proof}

Consider the contour
$$
\Gamma
=
\left\{ \exp\left(k^{-1} + it\right) : 0 \leq t \leq 2 \pi \right\}.
$$
We apply Lemma~\ref{l:contourint} and Jensen's inequality to get
\begin{align*}
P_r(N,k)
& =
\sum_{n > N} a(n,k) \\
& =
\sum_{n > N} \left| \left\langle \delta_n, \E^k \delta_{-1} \right\rangle \right|^2 \\
& \lesssim
\sum_{n > N} \int_0^{2 \pi}  \left| \left\langle \delta_n,
\left( \E - e^{1/k + i \theta} \right)^{-1} \delta_{-1}
\right\rangle \right|^2 \frac{d\theta}{2\pi},
\end{align*}
which proves \eqref{eq:rightbd:no:avg}. The proof of \eqref{eq:leftbd:no:avg} is nearly identical.

\end{proof}

With these pieces in place, we are now able to prove Theorem~\ref{t:dt08} very easily.

\begin{proof}[Proof of Theorem~\ref{t:dt08}]
Follow the proof of Theorem~\ref{DT07.thm7}, but replace \eqref{Pr} and \eqref{Pl} with Lemma~\ref{l:outbound:no:avg}. More specfically, we have
\begin{align*}
P_r(N,k)
& \lesssim
\int_0^{2 \pi} M_r\left(N, e^{1/k + i \theta} \right) \, \frac{d\theta}{2\pi} \\
& \lesssim
k^4  \int_0^{2\pi} \!
\left(
\max_{0 \leq n \leq N} \left\|Z\left(n, 0; e^{1/k + i \theta}\right) \right\|^2
\right)^{-1} \, \frac{d \theta}{2 \pi}.
\end{align*}
The first inequality is \eqref{eq:rightbd:no:avg}, and the second is Lemma~\ref{DT07.Lemma3}. The bound on $P_l$ works similarly.
\end{proof}

We can use this to prove the bound on the $\beta$'s from Theorem~\ref{DT07.thm1}.

\begin{proof}[Proof of Theorem~\ref{DT07.thm1}]

For $\psi = \delta_{-1}$, this follows immediately from Theorem~\ref{t:dt08}. More precisely, the assumption \eqref{RightInequality} together with Theorem~\ref{t:dt08} implies that $P_r(Ck^\gamma,k) + P_l(Ck^\gamma,k)$ goes to zero as $k \to \infty$ faster than $k^{-s}$ for any $s > 0$. Enlarging $C$ if necessary, we may assume that $a_\psi(n,k) = 0$ whenever $|n| > Ck+C$. This implies
\begin{align*}
|X|^p_\psi(k)
& =
\sum_{n \in \Z} \left(|n|^p + 1\right) a_\psi(n,k) \\
& =
\sum_{|n| \le Ck^\gamma}\left(|n|^p + 1\right) a_\psi(n,k)
+
\sum_{Ck^\gamma < |n| \le Ck+C} \left(|n|^p + 1\right)  a_\psi(n,k) \\
& \le
Ck^{\gamma p} + 1 + (Ck+C)^p\left( P_r(Ck^\gamma,k) + P_l(Ck^\gamma,k) \right).
\end{align*}
for all $k \in \Z_+$ and $p > 0$. Consequently, we obtain the theorem for $\psi = \delta_{-1}$. For $\psi = \delta_n$ with $n \in \Z$, the theorem follows by applying the case $\psi = \delta_{-1}$ to the CMV operator $S^{n+1} \E S^{-n-1}$, where $S: \ell^2(\Z) \to \ell^2(\Z)$ denotes the left shift. The result for finitely supported $\psi$'s is then obvious.

\end{proof}

\part{Examples}

\section{Delocalization for Quantum Walks in Polymers}\label{sec.6}

\subsection{Setting}

We describe a quantum walk analog of polymer models; compare \cite{JSS}. Fix a finite subset $\A \subseteq \mathrm{U}(2)$, and let $\A^* = \bigcup_{j=0}^\infty \A^j$ denote the free monoid over $\A$, that is, the set of all finite words which can be obtained by concatenating elements of $\A$. A \emph{polymer model} is given as soon as one chooses $n \ge 2$ and elements $u_1,\ldots,u_n \in \A^*$, which we call the \emph{basic chains} of the model. We can then define the family of quantum walks generated by these basic chains to be those whose sequences of coins may be obtained by concatenating elements of $\{u_1,\ldots,u_n\}$. More formally, let $\Omega = \{1,\ldots,n\}^{\Z}$, and, for each $\omega \in \Omega$, let $U_{\omega}$ be a quantum walk obtained by concatenating $\ldots,u_{\omega_{-1}}, u_{\omega_0}, u_{\omega_1},\ldots$ with the position of the origin normalized to be at the start of $u_{\omega_0}$. For each $j$ let $T_j(z)$ denote the Szeg\H{o} transfer matrix across the coin sequence $u_j$ at spectral parameter $z$. We say that a spectral parameter $z \in \partial \D$ is \emph{critical} for the random polymer model if the following three conditions hold:
\begin{enumerate}
\item $|\tr(T_j(z))| \le 2$ for every $1 \le j \le n$.
\item $T_j(z) T_k(z) = T_k(z) T_j(z)$ for all $1 \le j,k \le n$.
\item If $|\tr(T_j(z))| = 2$, then $T_j(z) = \pm I$.
\end{enumerate}

We distinguish such spectral parameters as ``critical'' in analogy with the critical energies for the Schr\"odinger equation with random polymer potential; in the context of the Schr\"odinger equation, these correspond to energies at which localization lengths diverge and Anderson Localization breaks down \cite{DT03},  \cite{DeBievreGerminet}, \cite{DunlapWuPhilips}, \cite{JSS}.
At each of these critical values, it is easy to produce a uniform bound on the associated transfer matrices; by general arguments \cite{DFLY2}, this ensures that the critical energy is in the spectrum.

Similarly, in our CMV operator setting any critical spectral parameter is contained in the spectrum of every operator of the form $\E_{\omega}$ with $\omega \in \Omega$.

\begin{ex}
It is not too hard to construct polymer models with critical spectral parameters. For example, notice that
$$
S(\alpha, 1) S(0,1) S(-\alpha,  1) S(0,1)
=
I
$$
for all $\alpha \in \D$. In particular, $z = 1$ is a critical spectral parameter for any random polymer model with basic chains of the form $(R_\theta,R_{-\theta})$ with $\theta \in \left( -\frac{\pi}{2}, \frac{\pi}{2} \right)$, where $R_\theta$ denotes the rotation matrix defined by
\[
R_\theta
=
\begin{pmatrix}
\cos \theta & -\sin \theta \\
\sin \theta &  \cos \theta
\end{pmatrix}.
\]

\bigskip

In Section~\ref{sec.8}, we will study the a quantum walk modelled on the Thue--Morse subshift. This can also be incorporated into the framework of polymer models, and we will see that this model has many critical spectral parameters.
\end{ex}

\begin{theorem} \label{t:poly:deloc}
Given a polymer model with a critical spectral parameter, one has $\widetilde\beta^\pm_{\delta_0}(p) \geq 1 - \frac{1}{p}$.
\end{theorem}

\begin{proof}
Let $z_0$ denote a critical parameter, and put $A_R = \{z_0\}$ for all $R>0$. Using the definition of criticality, we can choose $C_0 > 0$ so that \eqref{eq:szego:powerbounds} holds for all $R>0$ with $\gamma = 0$. Since $|B_K| = 2 K^{-1}$, we deduce $\left\langle |X|^p_{\delta_0} \right\rangle(K) \gtrsim K^{p - 1}$ from \eqref{eq:momentbound}. Thus, we get $\widetilde\beta^-_{\delta_0}(p) \ge 1 - \frac{1}{p}$.
\end{proof}

\section{The Fibonacci Quantum Walk} \label{sec:fqw}

\subsection{Setting}

In this section, we will discuss a time-homogeneous quantum walk on $\Z$ with two coins, whose distribution on the line is modulated by the Fibonacci subshift. We will apply both of our general theorems to deduce upper and lower bounds on the corresponding quantum walk. This is by far the most substantial application of our methods; it will require the most effort in order to verify the inputs needed to apply the general theorems. The key observation here is that one can restrict the range of parameters in such a way that the spectra of the canonical periodic approximants exhibit the same combinatorial structure that one sees in the Schr\"odinger case; compare \cite{Raymond}. Since our goal is to demonstrate the use of the transfer matrix method in quantum dynamics, we will freely make various simplifications throughout this section in order to increase clarity (without removing any of the essential obstacles). In particular, we do not attempt to locate optimal parameter ranges, nor do we attempt to find optimal constants in our proofs.

First, we recall how the Fibonacci sequence and subshift are generated. Consider an alphabet $\A$ with two symbols, $\A = \{a,b\}$, and let $\A^*$ denote the free monoid generated by $\A$, i.e., the set of all finite words obtained by concatenating elements of $\A$. The Fibonacci substitution $S$ sends $a$ to $ab$ and $b$ to $a$. Naturally, $S$ can be extended by concatenation to maps $\A^* \to \A^*$ and $\A^{\Z_+} \to \A^{\Z_+}$ which we will also denote by $S$. The \emph{Fibonacci sequence}, which we denote by $u_\Fib$, is the unique element of $\A^{\Z_+}$ that is invariant under $S$. It is the limit of the words $s_k := S^k(a)$ as $k \to \infty$ in an obvious sense. More precisely, with $s_0 = a$, $s_1 =ab$, $s_2= aba$, etc., we see that $s_k$ is a prefix of $s_{k+1}$ for all $k$, so we define $u_\Fib$ to be the unique sequence in $\A^{\Z_+}$ which has $s_k$ as a prefix for every $n \ge 0$:
$$
u_\Fib
=
abaababaabaab \ldots.
$$
The \emph{Fibonacci subshift} $\Omega_\Fib$ is then defined to be the dynamical hull of $u_\Fib$ in $\A^{\Z}$, the space of two-sided sequences over $\A$; more precisely,
$$
\Omega_\Fib
=
\left\{
\omega \in \A^{\Z} : \text{every finite subword of $\omega$ occurs in } u_\Fib
\right\}.
$$
To generate a quantum walk in a Fibonacci environment, choose $\theta_a , \theta_b \in \left(-\frac{\pi}{2}, \frac{\pi}{2} \right)$, and consider the rotations
$$
R_a
=
R_{\theta_a}
=
\begin{pmatrix}
\cos \theta_a & -\sin \theta_a \\
\sin \theta_a & \cos \theta_a
\end{pmatrix},
\quad
R_b
=
R_{\theta_b}
=
\begin{pmatrix}
\cos \theta_b & -\sin \theta_b \\
\sin \theta_b & \cos \theta_b
\end{pmatrix}.
$$
Given $\omega \in \Omega_\Fib$, we associate a sequence of coins $\{ Q_{\omega,n} \}_{n \in  \Z}$ via $Q_{\omega,n} = R_{\omega_n}$. The associated unitary operator will be denoted by $U_\omega$. Inspecting \eqref{e.correspondence} one sees that $U_\omega$ already has the form of an extended CMV matrix with Verblunsky coefficients
$$
\alpha_{2n+1}
=
\sin(\theta_{\omega_n}),
\quad
\alpha_{2n}
=
0, \text{ for all } n \in \Z.
$$
We will therefore denote $U_\omega$ by $\E_\omega$ to emphasize this fact. Notice also that we must avoid $\pm\pi/2$, for if either $\theta_a$ or $\theta_b$ is $\pm\pi/2$, then we have $|\alpha_{2n + 1}| = 1$ for infinitely many $n \in \Z$, so $\E_\omega$ decouples into a direct sum of finite blocks, which means that the dynamics are trivially localized. By a standard argument using minimality and strong operator approximation, there is a uniform compact set $\Sigma = \Sigma(\theta_a,\theta_b) \subseteq \partial \D$ such that $\Sigma = \sigma(\E_\omega)$ for every $\omega \in \Omega_\Fib$. Specifically, if $\omega_1, \omega_2 \in \Omega_\Fib$, then $\omega_1$ may be approximated by shifts of $\omega_2$ in the product topology, and thus, $\E_{\omega_1}$ may be approximated in the strong operator topology by operators which are unitarily equivalent to $\E_{\omega_2}$. One may then apply standard results (e.g.\ \cite[Theorem~VIII.24]{ReedSimon1}) to deduce $\sigma(\E_{\omega_1}) \subseteq \sigma(\E_{\omega_2})$; equality of spectra follows by symmetry (one may exchange the roles of $\omega_1$ and $\omega_2$ in the foregoing argument).

 To avoid cluttering the notation, we will suppress the dependence of the various objects in this section on the choice of $\theta_a$ and $\theta_b$.

We consider the initial state $\psi = \delta_0$ and study the spreading in space of $\mathcal{E}_\omega^n \psi$ as $|n| \to \infty$ with respect to this basis. For the sake of clarity and simplicity, we will focus on the dynamics generated by the element $\omega_0 \in \Omega_\Fib$ obtained by applying $S^2$ iteratively to the germ $b|a$, where the vertical bar separates sites $-1$ and $0$. We may write $\omega_0 = u_\Fib^R ab| u_\Fib$, where $u^R$ denotes the reversal of the word $u$. Thus, $\omega_0$ agrees with $u_\Fib$ on the nonnegative semiaxis, and it has the reflection symmetry
\begin{equation} \label{eq:fib:refl}
\omega_0(-n)
=
\omega_0(n-3)
\text{ for all } n \in\Z\setminus\{1,2\}.
\end{equation}
Equivalently, one may write
\begin{equation} \label{eq:omega0}
\omega_0(n)
=
\begin{cases}
a & \left\{ (n+1) \varphi^{-1} \right\} \in \left[ 1 - \varphi^{-1}, 1\right) \\
b & \left\{ (n+1) \varphi^{-1} \right\} \in \left[ 0, 1 - \varphi^{-1} \right)
\end{cases},
\end{equation}
where $\{x\} = x - \lfloor x \rfloor$ denotes the fractional part of $x \in \R$ and $\varphi = \frac{\sqrt 5 + 1}{2}$ denotes the golden ratio. One can deduce \eqref{eq:omega0} from \cite[Lemma~1]{BIST}, for example.

The structure of the remainder of the section is as follows: In Subsection~\ref{ssec:trace}, we describe the trace-map formalism for the operators $\E_\omega$, and we state our dynamical bounds precisely in terms of quantities associated to the trace map. Subsection~\ref{ssec:fibproof} contains a complete proof of the main theorem, modulo a pair of lemmas which contain some technical estimates on the trace map and on the growth of the Gesztesy--Zinchenko cocycle for this model. Subsection~\ref{ssec:fbands} works out a version of Raymond's combinatorial analysis of the spectrum in this setting. Finally, in Subsection~\ref{ssec:fibbounds}, we prove the technical estimates from Lemmas~\ref{l:dtkklbounds} and Lemma~\ref{l:dt:it}.

\subsection{Trace-Map Formalism and Fricke-Vogt Invariant} \label{ssec:trace}

In this section, we fix $\alpha = \alpha_{\omega_0}$ and $\E_0 = \E_{\alpha_{\omega_0}}$. The renormalization map for the Fibonacci sequence will play a key role in our analysis. More concretely, define
$$
M_n(z)
=
Z(2F_n , 0 ; z),
\; n \geq 0, \; z \in \C \setminus \{0\},
$$
where $Z$ denotes the Gesztesy--Zinchenko cocycle as before, and $F_n$ denotes the $n$th Fibonacci number, normalized by $F_{-1} = F_0 = 1$. We also define
$$
M_{-1}(z)
=
\sec\theta_b
\begin{pmatrix}
z & - \sin\theta_b \\
-\sin\theta_b & z^{-1}
\end{pmatrix}, \;
z \in \C \setminus \{0\}.
$$
The hierarchical structure of the Fibonacci substitution word produces a recursive relationship amongst these matrices, viz.
\begin{equation} \label{eq:transmat:renorm}
M_{n+1}(z)
=
M_{n-1}(z) M_n(z)
\text{ for all } n \ge 0 \text{ and all } z \in \C.
\end{equation}
This recursion leads to a number of nice consequences, which are well-known and not hard to prove; for the corresponding statements and proofs in the context of the discrete Schr\"odinger equation, consult \cite[Proposition~1(ii)]{Suto1987}. By the Cayley-Hamilton Theorem, \eqref{eq:transmat:renorm} leads to a recursion amongst the half-traces $x_n := \frac{1}{2} \tr(M_n(z))$, namely,
\begin{equation} \label{eq:trace:renorm}
x_{n+2}
=
2 x_n x_{n+1} - x_{n-1}
\text{ for all } n \ge 0.
\end{equation}
As a consequence of this recursion, the traces have a first integral given by the so-called Fricke--Vogt invariant. More precisely, if we define
$$
I(u,v,w)
=
u^2 + v^2 + w^2 - 2uvw - 1,
$$
then
\begin{equation} \label{eq:trace:inv}
I(x_{n+1}, x_n, x_{n-1})
=
I(x_{n+2}, x_{n+1}, x_n)
\text{ for all } n \ge 0.
\end{equation}
Consequently, by abusing notation, we will write $I(z) := I(x_{k+1}(z), x_k(z), x_{k-1}(z))$ for $k \ge 0$ and $z \in \C$. Evidently, $M_k(z)$ is the monodromy matrix for a periodic operator whose sequence of coins is defined by the sequence $\omega_k \in \A^{\Z}$ which repeats $s_k$ periodically. By Floquet theory, the spectra of these periodic operators are given by
\begin{equation}\label{sigma_k}
\sigma_k
:=
\{ z \in \partial \D : |x_k(z)| \le 1 \},
\;
k \ge -1.
\end{equation}
For more details on Floquet theory and the spectral characteristics of CMV matrices with periodic coefficients, the reader is referred to \cite[Section~11.2]{S2}. In this setting, the role of the coupling constant may be played by
\begin{equation} \label{eq:coup:def}
\mu
=
\mu(\theta_a,\theta_b)
:=
\inf_{k\ge-1} \min_{z \in \sigma_k} I(z).
\end{equation}
Let us also define
\begin{equation}\label{eq:fqw:kappadef}
\kappa
=
\kappa(\theta_a,\theta_b)
=
|\sec\theta_a\tan\theta_b - \tan\theta_a\sec\theta_b|
=
|\sec\theta_a||\sec\theta_b||\sin\theta_a - \sin\theta_b|.
\end{equation}
We can now state our main theorem, which describes the dynamics defined by $\E_0$ in the quantum walk analog of the large coupling limit.

\begin{theorem} \label{t:fqw:larginv}
Let $\pi/4 < \theta_b < \pi/2$ be given. There exist constants $m = m(\theta_b)$, $M = M(\theta_b)$, and $\lambda = \lambda(\theta_b)$ such that if $\theta_b < \theta_a < \pi/2$ with $\mu \ge \lambda$, one has
\begin{equation} \label{eq:fqwbetas}
\frac{1}{1+\tau} - \frac{3\tau+\eta}{p(1+\tau)}
=
\frac{p-3\tau-\eta}{p(1+\tau)}
\leq
\widetilde\beta^-_{\delta_0}(p)
\leq
\beta^+_{\delta_0}(p)
\leq
\frac{2\log\varphi}{\log\xi},
\end{equation}
where
\begin{equation}\label{def:zetaeta}
\xi = m\sqrt\mu,
\quad
\Xi = M\sqrt\mu,
\quad
\eta
=
\frac{\log\Xi}{\log\varphi} - 1,
\quad
\tau
=
\frac{2\log\left((\kappa + 2)(2\kappa + 5)^2\right)}{\log\varphi},
\end{equation}
and $\kappa$ is as in \eqref{eq:fqw:kappadef}. In particular, $\beta^+_{\delta_0}(p)$ goes to zero at least as fast as $\text{constant}/\log\sqrt\mu$ as $\mu \to \infty$.
\end{theorem}

\begin{remark} A few remarks are in order.

\begin{enumerate}
\item The conclusion of the theorem likely still holds with $\mu = \min_{z \in \Sigma} I(z)$. However, our choice of $\mu$ makes the proofs somewhat easier, so we use \eqref{eq:coup:def} to make the exposition more digestible.
\item Our method of proof for the lower bounds can allow $0 < \theta_b \le \pi /4$, but the proof of upper bounds breaks at $\theta_b = \pi/4$. However, on heuristic grounds, one may still expect that $\beta \searrow 0$ as $\theta_a \nearrow \pi/2$ for other values of $\theta_b$.
\item  For any $k \ge 0$, recall that $s_k = S^k(a)$; additionally, define $s_k' := S^k(b)$. By \cite[Lemma~3.2]{DamanikLenz}, every $\omega \in \Omega_\Fib$ may be decomposed as a concatenation of subfactors of the form $s_k$ and $s_k'$ in a unique way; this is known as the $k$-\emph{partition} of $\omega$. Theorem~\ref{t:fqw:larginv} can be extended to arbitrary elements of $\Omega_{\Fib}$ by analyzing various possibilities for the $k$-partition near the origin, but we will not present the details, since this clutters the presentation without adding substantial new content.
\end{enumerate}

\end{remark}

\subsection{Proof of Theorem~\ref{t:fqw:larginv}} \label{ssec:fibproof}

First, for Theorem~\ref{t:fqw:larginv} to have any interesting content, we need to verify that $\mu$ may indeed be made arbitrarily large by suitably choosing $\theta_a$ and $\theta_b$. To that end, let us begin by computing $I$ as a function of $z$. This is done in \cite{DMY2}, but there are a couple of typos and a sign error, and, most importantly, we have a substantially simpler expression for $I(z)$. For these reasons, we produce a corrected calculation here for the official record.\footnote{Note that \cite{DMY2} uses the Szeg\H{o} transfer matrices $\widetilde M_n(z) = z^{-F_n} T(2F_n,0;z)$, and we use the GZ transfer matrices. In light of \eqref{eq:szego:gz:rel}, this does not affect the value of $x_k(z)$ for $z \in \partial \D$, and therefore does not affect $I(z)$ for such $z$.} We have
\begin{align*}
M_0(z)
& =
\sec\theta_a
\begin{pmatrix}
z & - \sin\theta_a \\
-\sin\theta_a & z^{-1}
\end{pmatrix} \\
M_1(z)
& =
\sec\theta_a \sec\theta_b
\begin{pmatrix}
z^2 + \sin\theta_a \sin\theta_b
& * \\
*
& \sin\theta_a \sin\theta_b + z^{-2}
\end{pmatrix}.
\end{align*}
Thus, for $z \in \partial \D$, we have
\begin{align}
x_{-1}(z)
& =
\mathrm{Re}(z)\sec\theta_b, \quad
x_0(z)
=
\mathrm{Re}(z)\sec\theta_a, \\
\label{eq:fibx1}
x_1(z)
& =
\mathrm{Re}(z^2) \sec\theta_a \sec\theta_b
+
\tan\theta_a \tan\theta_b.
\end{align}
Consequently,
\begin{align*}
I(z)
 & =
\mathrm{Re}(z)^2(\sec^2\theta_a + \sec^2\theta_b)
+
\left( \mathrm{Re}(z^2) \sec\theta_a \sec\theta_b + \tan\theta_a \tan\theta_b \right)^2 \\
&
- 2\mathrm{Re}(z)^2 \sec\theta_a \sec\theta_b \left( \mathrm{Re}(z^2) \sec\theta_a \sec\theta_b + \tan\theta_a \tan\theta_b \right) - 1
\end{align*}
for all $z \in \partial \D$.

\begin{lemma} \label{l:fqw:inv}
For every $z \in \partial \D$, we have
\begin{equation} \label{eq:fqw:inv}
I(z)
=
\kappa^2 ( \mathrm{Im}(z))^2.
\end{equation}
\end{lemma}

\begin{proof}
Let $z = e^{it} \in \partial \D$ be given. To compactify notation a bit, denote
\begin{align*}
c_a = \cos\theta_a, \quad
\rho_a = \sec\theta_a, \quad
s_a & = \sin\theta_a, \quad
t_a = \tan\theta_a \\
c_b = \cos\theta_b, \quad
\rho_b  = \sec\theta_b, \quad
s_b & = \sin\theta_b, \quad
t_b = \tan\theta_b \\
c = \mathrm{Re}(z) = \cos t, \quad
s = \mathrm{Im}(z) & = \sin t, \quad
c_2 = \mathrm{Re}(z^2) = \cos(2t).
\end{align*}
We have seen above that
$$
I(z)
=
c^2(\rho_a^2 + \rho_b^2)
+ \left(c_2 \rho_a \rho_b + t_a t_b \right)^2
- 2c^2 \rho_a^2 \rho_b^2 \left( c_2 + s_a s_b \right)
- 1.
$$
After liberally using trigonometric identities, one can deduce \eqref{eq:fqw:inv}. The detailed calculation follows:
\begin{align*}
I
& =
c^2(\rho_a^2 + \rho_b^2)
+ \left(c_2 \rho_a \rho_b + t_a t_b \right)^2
- 2c^2 \rho_a^2 \rho_b^2 \left( c_2 + s_a s_b \right)
- 1 \\
& =
\rho_a^2 \rho_b^2
\left( c^2(c_a^2 + c_b^2) + (c_2 + s_a s_b)^2 - 2c^2 c_2 - 2c^2 s_a s_b
\right) - 1 \\
& =
\rho_a^2 \rho_b^2
\left( c^2(c_a^2 + c_b^2) -c_2 + s_a^2 s_b^2 - 2s^2 s_a s_b
\right) - 1 \\
& =
\rho_a^2 \rho_b^2
\left( (1-s^2)(c_a^2 + c_b^2) + 2s^2 - 1 + (1-c_a^2)(1 - c_b^2) - 2s^2 s_a s_b
\right) - 1 \\
& =
\rho_a^2 \rho_b^2 \left(s^2 (s_a^2 + s_b^2) - 2s^2 s_a s_b + c_a^2 c_b^2 \right) - 1 \\
& =
\rho_a^2 \rho_b^2 s^2(s_a - s_b)^2 \\
& =
( \mathrm{Im}(z))^2
\left( \sec\theta_a \tan\theta_b - \tan\theta_a \sec\theta_b \right)^2,
\end{align*}
which proves \eqref{eq:fqw:inv}. The third line is obtained by expanding the binomial square and using the identities $c_2 = 2c^2 - 1 = c^2 - s^2$. The fourth line uses $c_2 = 1-2s^2$ and the Pythagorean identity. The fifth line is simple algebra (and uses the Pythagorean identities some more).
\end{proof}

\begin{remark}
From the expression in \eqref{eq:fqw:inv}, we make several observations.
\begin{enumerate}
\item The invariant is always nonnegative on the spectrum. Indeed, for any $z \in \partial \D$, $I(z) \ge 0$.
\item One cannot make the invariant uniformly large on $\partial \D$ by suitably choosing $\theta_a$, $\theta_b$. Indeed, one always has $I(1) = I(-1) = 0$.
\item For any fixed $\theta_b$, one has
$$
\lim_{\theta_a \nearrow \pi/2}
\kappa(\theta_a,\theta_b)
=
\infty.
$$
Thus, for fixed $\theta_b$, we would like to think of $\theta_a \approx \pi/2$ as the quantum walk analog of the large coupling regime. We will see later that one can force the spectrum away from $\pm 1$ for some parameters, which allows us to make the invariant uniformly large on $\Sigma$.
\item Similarly, we think of $\theta_a \approx \theta_b$ as the ``small coupling regime,'' since we can make $I$ uniformly small on $\partial \D$ by making $\theta_a$ and $\theta_b$ sufficiently close.
\end{enumerate}
\end{remark}

Let us note that we can characterize the uniform spectrum of $\E_\omega$ as the dynamical spectrum, i.e., the set of complex numbers at which the associated trace map has a bounded orbit. More precisely:

\begin{prop}
Fix $\theta_a, \theta_b \in \left( -\frac{\pi}{2}, \frac{\pi}{2} \right)$, let $\Sigma$ denote the uniform spectrum of every $\E_\omega$ with $\omega \in \Omega_\Fib$, and put
$$
B_\infty
=
\left\{
z \in \C : (x_n(z))_{n = -1}^\infty \text{ is a bounded sequence}
\right\}.
$$
Then $\Sigma = B_\infty$.
\end{prop}

\begin{proof}
The inclusion $\Sigma \subseteq B_\infty$ is in \cite[Lemma~2.3]{DMY2}. The other inclusion can be proved using ideas from Iochum--Testard \cite{IT91}. Namely, given $z \in B_\infty$, one can linearize the matrix recursion as in \cite{IT91} and use the assumption of bounded traces to derive a power-law upper bound on the norms of the matrices $\|Z(n,0;z)\|$ with $n \in \Z$. Indeed, this proof is nearly identical to the proof of Lemma~\ref{l:dt:it}. Once one has a power-law bound on the growth of the transfer matrices, the $z$ in question belongs to the spectrum by general principles (e.g., \cite[Theorem~6]{DFLY2}).
\end{proof}

In order to decide whether or not a given $z \in \partial \D$ is a member of $B_\infty$, it is helpful to have a computable condition which guarantees escape of the corresponding trace orbit. The following statement is not hard to prove; it is \cite[Lemma~4.3]{DGLQ} applied to the special case $\alpha = \varphi^{-1}$.

\begin{lemma} \label{l:escape}
Fix $z \in \C$ and $\delta \ge 0$. If there exists $k_0 \ge 0$ such that
\begin{equation} \label{eq:escape:cond}
|x_{k_0}(z)| > 1 + \delta,
\quad
|x_{k_0 + 1}(z)| > 1 + \delta,
\quad
\text{and}
\quad
|x_{k_0 + 1}(z)| > |x_{k_0-1}(z)|,
\end{equation}
then the sequence $(x_k(z))_{k=-1}^\infty$ is unbounded. Moreover, if \eqref{eq:escape:cond} holds, then $|x_{k_0 + j}(z)| \ge (1+\delta/2)^{F_{j-1}}$ for all $j \ge 0$.
\end{lemma}

We are now able to verify that the Fricke--Vogt invariant can be made uniformly large on the spectrum in the large coupling regime.

\begin{prop} \label{p:largecoup}
For any $\lambda >0$, there exist choices of $\theta_a, \theta_b > 0$ such that $\mu \ge \lambda$. More specifically, for any $0 < \theta_b < \frac{\pi}{2}$, there exists $\phi_0 = \phi_0(\theta_b) > \theta_b$ such that $I(z) \geq \lambda$ for all $z \in \sigma_k$ and all $k \ge -1$ whenever $\phi_0 < \theta_a < \frac{\pi}{2}$.
\end{prop}

\begin{proof}
We will show that we can force the spectrum into the cone where the invariant is large by forcing the escape condition to hold for $k_0 = 0$ and $\delta = 0$ in a suitable cone containing the real axis. Fix $\theta_b > 0$, and suppose $\theta_a > \theta_b$ is large enough that
\begin{align} \label{eq:coupling:Ibound}
\kappa^2
=
(\sec\theta_a \tan\theta_b - \tan\theta_a \sec\theta_b)^2
& \ge
\frac{\lambda}{\sin^2\theta_b} \\
\label{eq:coupling:angbound}
(2\cos^2\theta_b - 1) + \sin\theta_a \sin\theta_b
& >
\cos\theta_a.
\end{align}
Notice that the second condition can be satisfied by taking $\theta_a$ sufficiently close to $+\pi/2$, since
$$
2\cos^2\theta_b - 1 + \sin\theta_b
>
\cos^2\theta_b
>
0.
$$
Now, recall:
\begin{align*}
x_{-1}(z)
& =
\mathrm{Re}(z)\sec \theta_b,
\quad
x_0(z)
=
\mathrm{Re}(z) \sec \theta_a, \\
x_1(z)
& =
\mathrm{Re}(z^2)
\sec\theta_a \sec\theta_b + \tan \theta_a \tan \theta_b.
\end{align*}
It is easy to see that
$$
|x_{-1}(z)| , \, |x_0(z)| > 1
$$
whenever $|\mathrm{Re}(z)| > \cos\theta_b$, since $\theta_a > \theta_b > 0$. We also want to show $|x_1| > |x_{-1}|$, which holds whenever
\begin{equation} \label{eq:Rez2bound}
|\mathrm{Re}(z^2) + \sin\theta_a \sin\theta_b|
>
|\mathrm{Re}(z)|\cos\theta_a.
\end{equation}
By \eqref{eq:coupling:angbound}, we see that \eqref{eq:Rez2bound} holds for all $z$ with $|\mathrm{Re}(z)| > \cos\theta_b$. Since $|x_{-1}(z)| > 1$ for such $z$, we also get $|x_1(z)| > 1$ for those values of $z$, so the escape condition \eqref{eq:escape:cond} holds for $k_0 = 0$, $\delta = 0$, and all $z \in \partial \D$ with $|\mathrm{Re}(z)| > \cos\theta_b$. It follows that
$$
\sigma_k
\subseteq
\sigma_{-1}
=
\{z \in \partial \D : |\mathrm{Im}(z)| \ge \sin\theta_b \}
\text{ for all } k \ge -1.
$$
Consequently, we have $I(z) \ge \lambda$ on $\sigma_k$ for all $k \ge -1$ by \eqref{eq:fqw:inv} and \eqref{eq:coupling:Ibound}.
\end{proof}

The proof of Theorem~\ref{t:fqw:larginv} relies on the following pair of technical lemmas. The first lemma establishes upper and lower bounds on the derivatives of the iterates of the trace map, while the second proves power-law bounds on the growth of the GZ cocycle on the periodic spectra; compare \cite[Lemma~5.2]{KKL}, \cite[Proposition~3.2 and Lemma~3.5]{DT03}. The proofs would distract from the overall narrative flow of the paper; thus we postpone them until the end of the section.

\begin{lemma} \label{l:dtkklbounds}
For each $\theta_b > \pi/4$, there exist constants $m = m(\theta_b)$, $M = M(\theta_b)$, and $\lambda = \lambda(\theta_b)$ such that the following holds. If $\theta_b < \theta_a < \pi/2$ and $\mu \ge \lambda$, then, for every $z \in \sigma_k$, we have
\begin{equation} \label{eq:xibounds}
\xi^{k/2}
\le
|x_k'(z)|
\le
\Xi^k,
\end{equation}
where $\xi =m\sqrt\mu$ and $\Xi = M \sqrt\mu$, as in the statement of Theorem~\ref{t:fqw:larginv}.
\end{lemma}

\begin{lemma} \label{l:dt:it}

For all $\theta_a$ and $\theta_b$ as in the statement of Theorem~\ref{t:fqw:larginv}, there is a constant $C > 0$ such that for every $n \in \Z$ with $-2F_k \le n \le 2F_k$ with $n \neq 0$, one has
\begin{equation} \label{eq:dt:it}
\| Z(n,0;z) \|
\le
C|n|^{\tau/2}
\end{equation}
for all $z \in \sigma_k$, where $\tau$ is as in \eqref{def:zetaeta}, and $Z$ denotes the Gesztesy--Zinchenko cocycle associated to $\E_0$. Consequently, we have
$$
\| Z(n,m;z) \|
\le
C^2|n|^{\tau}
$$
for all $z \in \sigma_k$ and all $n,m$ with $ |n| , |m| \leq 2F_k$ and $n \neq 0$.

\end{lemma}

\begin{proof}[Proof of Theorem~\ref{t:fqw:larginv}]
Fix $\theta_b$ and $\theta_a \in (\theta_b,\pi/2)$ sufficiently large. The theorem consists of two nontrivial inequalities: a lower bound on $\widetilde\beta^-$ and an upper bound on $\beta^+$. We will prove the lower bound using Theorem~\ref{t:dt:powerlaw} and the upper bound using Theorem~\ref{DT07.thm1}. Recall that $F_k \sim \varphi^k$.

\bigskip

\noindent \textbf{Lower Bound.} By Lemma~\ref{l:dtkklbounds} and the Mean Value Theorem, we have $|B_k| \ge 4 \Xi^{-k}$ for each band $B_k$ of $\sigma_k$.  Note that $F_k \sim \varphi^k$. As a consequence of the $2F_k$-periodicity of the corresponding CMV operator and standard results in the theory of periodic operators (see, for example, \cite[Theorem~11.1.1]{S2}), we know that $\sigma_k$ consists of $2F_k$ nondegenerate closed intervals, known as \emph{bands}. These facts imply that
$$
|\sigma_k|
\ge
8F_k \Xi^{-k}
\gtrsim
F_k^{-\eta},
$$
where $\eta = \frac{\log(\Xi)}{\log\varphi} - 1,$ as in \eqref{def:zetaeta}. Note that the implicit constant on the right hand side depends on $\Xi$, and hence on $\mu$ and $\theta_b$. Combining this lower bound on the measure of the periodic spectra with the polynomial bounds from Lemma~\ref{l:dt:it}, we get the desired conclusion from Theorem~\ref{t:dt:powerlaw} by taking $A_R = \sigma_{\ell(R)}$ for each $R > 1$, where $\ell(R)\in \Z_+$ is the unique integer with $F_{\ell-1} < R \le F_\ell$. Specifically, with $R_K = K^{\frac{1}{1+\tau}}$, one has
$$
|A_{R_K}|
=
|\sigma_{\ell(R_K)}|
\gtrsim
F_{\ell(R_K)}^{-\eta}
\gtrsim
R_K^{-\eta}
=
K^{\frac{-\eta}{1+\tau}},
$$
which, by \eqref{eq:momentbound}, implies
$$
\left\langle |X|^p_{\delta_0} \right\rangle(K)
\gtrsim
|A_{R_K}| K^{\frac{p-3\tau}{1+\tau}}
\gtrsim
K^{\frac{p-3\tau-\eta}{1+\tau}}.
$$
Consequently,
$$
\widetilde\beta^-_{\delta_0}(p)
\ge
\frac{p-3\tau-\eta}{p(1+\tau)},
$$
which proves the lower bound from \eqref{eq:fqwbetas}.

\bigskip

\noindent \textbf{Upper Bound.} Define
$$
\sigma_k^\delta
:=
\left\{ z \in \C : |x_k(z)| \leq 1+\delta \right\}.
$$
Using the lower bound from \eqref{eq:xibounds}, one can argue as in \cite[Proposition~3]{DT07} and use the Koebe Distortion Theorem to see that for every $\delta > 0$ small enough, there is a constant $C_\delta > 0$ such that
$$
\sigma_k^\delta
\subseteq
B\left(0,\exp\left(C_\delta \xi^{-k/2}\right) \right)
$$
for all $k \in \Z_+$. We now fix $\delta > 0$ small enough that $\sigma_n^\delta$ has $F_n$ connected components for all $n \in \Z_+$; the existence of such a $\delta$ follows from the argument which proves \cite[Lemma~5]{DT07}. Let $\nu > 0$ be given, and put
$$
\gamma
=
\gamma(\nu)
=
\frac{\log(\xi)}{2(1+\nu)\log(\varphi)}.
$$
Thus, for all $k \in \Z_+$ large enough (depending on $\nu$), we have
\begin{equation} \label{eq:escape:begins}
\sigma_k^\delta \cup \sigma_{k+1}^\delta
\subseteq
B\left(0, \exp\left( C_\delta F_k^{-\gamma} \right) \right).
\end{equation}
Let $K_0$ be large enough that $ze^{K_0^{-1}} \in \sigma_{-1}^\delta$ whenever $z \in \sigma_{-1}$. Given $K \ge K_0$, define $n = n(K)$ and $N = N(K)$ by
$$
\frac{F_{n-1}^\gamma}{C_\delta}
\leq
K
<
\frac{F_{n}^\gamma}{C_\delta},
\quad
N
=
2F_{n + \sqrt n},
$$
where we abbreviate $F_x := F_{\lfloor x \rfloor}$ for real numbers $x$. In particular, we have $K^{-1} > C_\delta F_n^{-\gamma}$, so, for complex numbers of the form $we^{1/K}$ with $w \in \sigma_{-1}$, the escape condition \eqref{eq:escape:cond} holds for some $k_0 \le n(K)$. On the other hand, \eqref{eq:escape:cond} already holds with $k_0 = \delta = 0$ for all $w \in \partial \D \setminus \sigma_{-1}$ by our choice of parameters, so, since escape is an open condition, we may enlarge $K_0$ (if necessary) so that \eqref{eq:escape:cond} still holds (with $k_0 = \delta = 0$) for all $z$ of the form $we^{K^{-1}}$ with $w \in \partial \D \setminus \sigma_{-1}$ and $K \ge K_0$. Consequently, for every $z$ with $|z| = e^{1/K}$, we have
$$
\left| x_q(z) \right|
\ge
\left( 1 + \frac{\delta}{2} \right)^{F_{q-n(K)-1}}
$$
for every integer $q \ge n(K)$, by Lemma~\ref{l:escape}. For every $\epsilon > 0$, our choice of $N(K)$ implies that there is a constant $c_\epsilon > 0$ such that
$$
N(K)
\le
c_\epsilon K^{\gamma^{-1} + \epsilon}
$$
for every $K \in \Z_+$. Putting everything together, we have
\begin{align*}
\max_{0 \le n \le c_\epsilon K^{\gamma^{-1} + \epsilon} } \left\|Z\left( n,0; e^{i\theta + 1/K} \right) \right\|
& \ge
\left\| Z\left( N(K), 0 ; e^{i\theta + 1/K} \right) \right\| \\
& \gtrsim
\left|x_{N(K)} \left(e^{i\theta + 1/K} \right) \right| \\
& \ge
\left(1 + \frac{\delta}{2} \right)^{F_{\sqrt{n(K)}-1}}
\end{align*}
for all $\theta \in [0, 2\pi)$, where we have used Lemma~\ref{l:escape} and $N(K) - n(K) = \sqrt{n(K)}$ to obtain the final line. We then have
\begin{align*}
\int_0^{2\pi} \! \left( \max_{0 \le n \le c_\epsilon K^{\gamma^{-1} + \epsilon} } \left\| Z(n,0; e^{i\theta + 1/K}) \right\|^2 \right)^{-1} \, \frac{d\theta}{2\pi}
& \lesssim
\left(1 + \frac{\delta}{2} \right)^{-2F_{\sqrt{n(K)}-1}}.
\end{align*}
The right hand side decays faster than any negative power of $K$, and we can prove a similar bound on the other half-line using the reflection symmetry \eqref{eq:fib:refl}. Thus, we have $\beta^+_{\delta_0}(p) \leq \gamma(\nu)^{-1} + \epsilon$ for all $p > 0$ by Theorem~\ref{DT07.thm1}. Since this holds for all $\epsilon > 0$ and $\nu > 0$, we obtain $\beta^+_{\delta_0}(p) \le \gamma(0)^{-1}$, which is the upper bound in \eqref{eq:fqwbetas}.
\end{proof}

\subsection{Band Combinatorics for Fibonacci Quantum Walks} \label{ssec:fbands}

Throughout this subsection, we adopt the standing assumptions $0 < \theta_b < \theta_a < \pi/2$ and $\mu \ge 32$. Notice that the assumption $\mu \ge 32$ implies $\sigma_{k-1} \cap \sigma_k \cap \sigma_{k+1} = \emptyset$ for all $k \ge 0$, since $I(x_{k-1}, x_k, x_{k+1}) \ge \mu$ on $\sigma_k$ for all $k$.

 Let us say that a band $B_k \subseteq \sigma_k$ is a \emph{type A band} if $B_k \subseteq \sigma_{k-1}$ and a \emph{type B band} if $B_k \subseteq \sigma_{k-2}$. The assumption on the $\theta$'s means that $\sigma_0$ consists of two bands of type A, and $\sigma_1$ consists of four bands of type B. The first claim is an obvious consequence of $0 < \theta_b < \theta_a < \pi/2$, but the second requires an argument, which is supplied by the following lemma.

\begin{lemma}
If $0 < \theta_b < \theta_a < \pi/2$, then $\sigma_{-1} \supseteq \sigma_0 \cup \sigma_1$. If, in addition, $\mu \ge 32$, then $\sigma_0 \cap \sigma_1 = \emptyset$.
\end{lemma}

\begin{proof}
As noted above, the assumptions on $\theta_a$ and $\theta_b$ immediately give $\sigma_0 \subseteq \sigma_{-1}$.  To prove the other inclusion, suppose $z \notin \sigma_{-1}$, i.e., $|\mathrm{Re}(z)| > \cos\theta_b$. Then the assumptions on the $\theta$'s imply
$$
\mathrm{Re}(z^2)
>
\cos(2\theta_b)
>
\cos(\theta_a + \theta_b)
=
\cos\theta_a\cos\theta_b - \sin\theta_a\sin\theta_b.
$$
Consequently,
\begin{align*}
x_1(z)
& =
\Re(z^2) \sec\theta_a\sec\theta_b + \tan\theta_a\tan\theta_b \\
& >
(\cos\theta_a\cos\theta_b - \sin\theta_a\sin\theta_b) \sec\theta_a\sec\theta_b + \tan\theta_a\tan\theta_b \\
& =
1,
\end{align*}
so $z \notin \sigma_1$. Thus $\sigma_1 \subseteq \sigma_{-1}$. With this in hand, we must have $\sigma_0 \cap \sigma_1 = \emptyset$, since $\mu \ge 32$.
\end{proof}

With the previous lemma in hand, we know that type A and B bands exhaust the periodic spectra at levels 0 and 1. To see that they exhaust the spectra at all levels, we use the following lemma, which is a CMV variant of \cite[Lemma~5.3]{KKL}. The proof is nearly identical to the proof in the Schr\"odinger case -- we include detailed arguments for the reader's convenience.

\begin{lemma} \label{l:bandtypes}
Assume $0 < \theta_b < \theta_a < \pi/2$ and $\mu \ge 32$. For every $k \ge 0$:
\begin{enumerate}
\item[{\rm(1)}] Every type A band of $\sigma_k$ contains a type B band of $\sigma_{k+2}$ and no other bands of $\sigma_{k+1}$ or $\sigma_{k+2}$.
\item[{\rm(2)}] Every type B band of $\sigma_k$ contains a type A band $B_{k+1} \subseteq\sigma_{k+1}$ and two type B bands from of $\sigma_{k+2}$ which sandwich $B_{k+1}$.
\end{enumerate}
\end{lemma}

\begin{proof}
\noindent \textbf{Case A.} Suppose $B_k \subseteq \sigma_k$ is a type-A band. By definition, $B_k \subseteq \sigma_{k-1}$, so we must have $B_k \cap \sigma_{k+1} = \emptyset$, since $\mu \ge 32$; equivalently, $|x_{k+1}| > 1$ on $B_k$. There is a unique $z_1 \in B_k$ such that $x_k(z_1) = 0$. Using the trace map \eqref{eq:trace:renorm}, we have
$$
|x_{k+2}(z_1)|
=
|x_{k-1}(z_1)|
\le
1,
$$
since $z_1 \in B_k \subseteq \sigma_{k-1}$. In particular, $B_k \cap \sigma_{k+2} \neq \emptyset$. Notice also that when $x_k = \pm 1$, we have
$$
|x_{k+2}|
\ge
2|x_{k+1}| - |x_{k-1}|
>
1,
$$
so any band of $\sigma_{k+2}$ which intersects $B_k$ must be entirely contained within $B_k$. Finally, suppose $\widetilde B_{k+2}$ is a band of $\sigma_{k+2}$ contained in $B_k$. By the IVT, $\widetilde B_{k+2}$ contains a point $z_2$ for which $x_{k+2}(z_2) + x_{k-1}(z_2) = 0$, so $x_k(z_2) x_{k+1}(z_2) = 0$ by \eqref{eq:trace:renorm}. Since $z_2 \notin \sigma_{k+1}$, $x_k(z_2) = 0$. Thus, there is a \emph{unique} band of $\sigma_{k+2}$ contained in $B_k$.

\bigskip

\noindent \textbf{Case B.} Suppose $B_k \subseteq \sigma_k$ is a type-B band, i.e.\ $B_k \subseteq \sigma_{k-2}$; immediately, one has $B_k \cap \sigma_{k-1} = \emptyset$. As above, if $x_k(z_3) = 0$, then
$$
|x_{k+1}(z_3)|
=
|x_{k-2}(z_3)|
\le
1,
$$
whence $B_k \cap \sigma_{k+1} \neq \emptyset$. As in the previous case, $|x_k| = 1$ forces $|x_{k+1}| > 1$, so any band of $\sigma_{k+1}$ which meets $B_k$ must be contained entirely within the interior of $B_k$. Moreover, running the argument above again, any band of $\sigma_{k+1}$ contained within $B_k$ must contain a spectral parameter at which $x_k$ vanishes, and so there is exactly one band of $\sigma_{k+1}$ contained in $B_k$. Now we consider bands of $\sigma_{k+2}$ which meet $B_k$. Notice that one may iterate the trace recursion to obtain
\begin{equation} \label{eq:tracerec:2}
x_{k+2}
=
2x_{k+1} x_k - x_{k-1}
=
(4 x_k^2 -1) x_{k-1} - 2 x_k x_{k-2}.
\end{equation}
As before, any band of $\sigma_{k+2}$ which meets $B_k$ is completely contained in $B_k$; specifically, if $|x_k| = 1$, then \eqref{eq:tracerec:2} forces $|x_{k+2}| > 3-2 = 1$. On the other hand, if $x_k = \pm 1/2$, then \eqref{eq:tracerec:2} forces $|x_{k+2}| = |x_{k-2}| \le 1$. Consequently, $B_k$ contains \emph{at least} two bands of $\sigma_{k+2}$ (note that we can't have a single large band, becaue $\sigma_{k+2}$ cannot meet $\sigma_{k+1}$, which contains the zero of $x_k$ in $B_k$, which separates the points where $x_k = \pm 1/2$). It remains to see that we have no more than two bands of $\sigma_{k+2}$ in $B_k$. Notice that
$$
(2x_k \pm 1)(x_{k+2} \pm x_{k-2})
=
(4x_k^2 - 1)(x_{k+1} \pm x_{k-1}).
$$
Now, fix a band $B_{k+2} \subseteq B_k$ of $\sigma_{k+2}$.  We may choose a fixed sign so that $x_{k+1} \pm x_{k-1}$ never vanishes on $B_{k+2}$. By the IVT, there exists $z_4 \in B_{k+2}$ so that $x_{k+2}(z_4) \pm x_{k-2}(z_4) = 0$. We must then have $x_k(z_4) = \pm 1/2$. Since there are only two points in $B_k$ where this happens, we are done -- there cannot be more than two bands of $\sigma_{k+2}$ in $B_k$.
\end{proof}

Lemma~\ref{l:bandtypes} implies that every band of $\sigma_k$ is of type A or B for every $k \ge 0$. Moreover, it also implies that $\sigma_k$ consists of $2F_k$ disjoint bands for all $k$, that is, $\sigma_k$ has no closed gaps. Finally, it also implies that $\sigma_{-1} \supseteq \sigma_k$ for all $k$, and hence $\sigma_{-1} \supseteq \Sigma$ by a strong approximation argument.

\subsection{Power-Law Bounds on the Growth of Fibonacci--Szeg\H{o} Cocycles} \label{ssec:fibbounds}

In this subsection, we prove Lemmas~\ref{l:dtkklbounds} and \ref{l:dt:it}. Throughout, we have the standing assumptions that $\pi/4 < \theta_b < \theta_a < \pi/2$ and that $\mu$ is large enough, where ``large enough'' is a $\theta_b$-dependent statement which becomes harder and harder to satisfy as $\theta_b \searrow \pi/4$. We do not attempt to find optimal constants, and instead opt for lucidity of presentation. Notice that
$$
u^2 + v^2 + w^2 - 2uvw - 1 = I
$$
implies that
$$
w
=
uv \pm \sqrt{I + (1-u^2)(1-v^2)}
$$
via the quadratic formula. This motivates us to define
$$
g_\pm(u,v,I)
=
uv \pm \sqrt{I + (1-u^2)(1-v^2)}.
$$

\begin{lemma} \label{l:kklderivbounds}
For $I \ge 4$ and $|u|,\, |v| \le 1$, we have
$$
\left| \frac{\partial g_\pm}{\partial u}(u,v,I) \right|, \;
\left| \frac{\partial g_\pm}{\partial v}(u,v,I) \right|
\le
1.
$$
Moreover,
$$
\left|\frac{\partial g_\pm}{\partial I}(u,v,I) \right|
\le
\frac{1}{2\sqrt{I}}.
$$
\end{lemma}

\begin{proof}
Since $g_\pm(u,v, I) = \frac{1}{2} f_\pm(2u,2v,2\sqrt{I})$ in Killip--Kiselev--Last's notation, the first two inequalities are immediate from \cite[Lemma~5.4]{KKL} and the chain rule. In particular, note that $I = \lambda^2/4$ when passing between our notation and theirs. The third bound is obvious.
\end{proof}

\begin{lemma} \label{l:tracederivbounds}
Suppose $\pi/4 < \theta_b < \pi/2$ is given. There exist constants $m = m(\theta_b)$, $M = M(\theta_b)$, and $\lambda = \lambda(\theta_b)$ such that if $\theta_b < \theta_a < \pi/2$ and $\mu \ge \lambda$, then the following statements hold with $\xi = m\sqrt\mu$ and $\Xi = M\sqrt\mu$. We have
$$
\xi
\le
\left| \frac{x_{k+1}'(z)}{x_k'(z)} \right|
\le
\Xi
$$
for all $z \in B_{k+1}$ if $B_{k+1}$ is a type A band, and
$$
\xi
\le
\left| \frac{x_{k+2}'(z)}{x_k'(z)} \right|
\le
\Xi
$$
for all $z \in B_{k+2}$ if $B_{k+2}$ is a type B band.
\end{lemma}

\noindent \textit{Notation.} We adopt the following notation, just over the course of the following proof. Given quantities $f$ and $g$ which depend on $\theta_a$ and $\theta_b$ (subject to all of our various standing assumptions thereupon), we denote $f \sim g$ if for every $\theta_b$, there exists $C = C(\theta_b)$ with $C^{-1} f \le g \le Cf$.

\begin{proof}
 One can follow the proofs of \cite[Lemma~3.5]{DT03} and \cite[Lemma~5.5]{KKL} and prove these statements by induction, but there are two new difficulties. First, the initial conditions are different in the CMV case, so we must verify that the induction is well-founded. Once this is done, we have to deal with extra terms which arise from the failure of $I$ to be constant on the spectrum. To deal with the first base case, observe that
$$
\left| \frac{x_0'(z)}{x_{-1}'(z)} \right|
=
\frac{\sec\theta_a}{\sec\theta_b}
\sim
\sqrt \mu
$$
for all $z \in \partial \D$. Next, for all $z \in \sigma_1$, we have
$$
-\cos(\theta_a-\theta_b)
\le
\mathrm{Re}(z^2)
\le
\cos(\theta_a + \theta_b)
<
\cos(2\theta_b)
<
0
$$
by \eqref{eq:fibx1} and standard trigonometric identities. Consequently, $|\mathrm{Im}(z)| \ge \sin\theta_b$ and $|\mathrm{Im}(z^2)| \ge \sin(\theta_a-\theta_b)$, so we obtain
\begin{align*}
\left| \frac{x_1'(z)}{x_{-1}'(z)} \right|
& =
\left| \frac{(2z-2z^{-3})\sec\theta_a\sec\theta_b}{(1-z^{-2})\sec\theta_b} \right| \\
& =
2 \left| \frac{\mathrm{Im}(z^2)}{\mathrm{Im}(z)} \right| \sec\theta_a \\
& \sim
\sqrt\mu.
\end{align*}
Now, we proceed inductively. Notice that the base case (and then each inductive step thereafter) implies $|x_k'| > 1$ on $\sigma_k$ as long as $\mu$ is large enough. Taking $\lambda(\theta_b)$ sufficiently large, we have $|x_k'(z)| > 1$ on $\sigma_k$ for $k = -1, 0$.

In the argument below, we will need to estimate $(\partial g_\pm/\partial I) I'(z)$ for $z \in \sigma_k$. Since $|I'(z)| = |\mathrm{Re}(z)|\kappa^2$ and $|\mathrm{Im}(z)| \ge\sin\theta_b$ for all $z\in\sigma_k$, we have
\begin{equation} \label{eq:ideriv:bound}
\left| \frac{\partial g_\pm}{\partial I} I'(z) \right|
\le
\frac{1}{2\sqrt\mu} \cdot \frac{\mu}{\sin^2\theta_b}
=
\frac{\sqrt\mu}{2\sin^2\theta_b}
\end{equation}
for all $z \in \sigma_k$ and all $ k \ge -1$.

\bigskip

\noindent \textbf{Case 1: Type A Band.} Suppose $B_{k+1}$ is a type-A band of $\sigma_{k+1}$ for some $k \ge 1$; let $B_k$ denote the band of $\sigma_k$ with $B_{k+1} \subseteq B_k$. Consequently, we have $B_k \cap \sigma_{k-1} = \emptyset$ and $B_k \subseteq \sigma_{k-2}$, so $B_k$ is a type B band. Differentiating the renormalization map \eqref{eq:trace:renorm} and dividing by $x_k'$, we get
$$
\frac{x_{k+1}'}{x_k'}
=
2x_{k-1} + \frac{2 x_k x_{k-1}'}{x_k'} - \frac{x_{k-2}'}{x_k'}.
$$
To estimate the second term on the right hand side, note that $x_{k-1} = g_\pm(x_k, x_{k-2},I)$, so
$$
x_{k-1}'
=
\frac{\partial g_\pm}{\partial u} x_{k-2}'
+
\frac{\partial g_\pm}{\partial v} x_{k}'
+
\frac{\partial g_\pm}{\partial I} I',
$$
which implies
$$
\left|\frac{x_{k-1}'}{x_k'} \right|
\le
\left|\frac{x_{k-2}'}{x_k'}\right| + 1 + \frac{\sqrt\mu}{2|x_k'|\sin^2\theta_b}
\le
2 + \frac{\sqrt\mu}{2\sin^2\theta_b}
$$
on $B_{k+1}$, where we have used \eqref{eq:ideriv:bound}, $B_{k+1} \subseteq B_k \subseteq \sigma_{k-2}$, Lemma~\ref{l:kklderivbounds}, and the inductive hypothesis. Consequently,
\begin{equation} \label{eq:KKL1}
\left| \frac{2 x_k x_{k-1}'}{x_k'} - \frac{x_{k-2}'}{x_k'} \right|
\le
5 + \frac{\sqrt\mu}{\sin^2\theta_b}.
\end{equation}
Since $|x_k|, |x_{k-2}| \le 1$ on $B_{k+1}$, one can use $x_{k-1} = g_\pm(x_k,x_{k-2},I)$ to see that
$$
\sqrt\mu - 1
\le
|x_{k-1}|
\le
\sqrt{\mu} + 2
$$
on $B_{k+1}$.
Combining this with \eqref{eq:KKL1} and using the triangle inequality twice, we obtain
$$
\left(2 - \frac{1}{\sin^2\theta_b} \right) \sqrt{\mu} - 7
\le
\left|\frac{x_{k+1}'}{x_k'}\right|
\le
\left(2 + \frac{1}{\sin^2\theta_b}\right)\sqrt{\mu} + 9
$$
on $B_{k+1}$. Note that one needs $\theta_b > \pi/4$ for the lower bound to have any significance, even for large $\mu$.

\bigskip

\noindent \textbf{Case 2: Type B Band.} Let us suppose that $B_{k+2}$ is a type-B band of $\sigma_{k+2}$ for some $k \ge 0$; Let $B_k$ denote the band of $\sigma_k$ which contains $B_{k+2}$. There are two subcases to consider. (Notice that $k=0$ falls in Subcase~2ii, so one need not worry about the $x_{-2}$'s that might appear in Subcase~2i).

\bigskip

\noindent \textbf{Subcase 2i: $B_k$ is type B.} This implies $B_{k+2} \subseteq B_k \subseteq \sigma_{k-2}$ and $B_k \cap \sigma_{k-1} = \emptyset$. As in the previous argument, we have
$$
\sqrt\mu - 1
\le
|x_{k+1}|
\le
\sqrt{\mu} + 2
$$
on $B_{k+2}$, and we may estimate
\begin{equation} \label{eq:subcase2iderivbound}
\left|\frac{x_{k-1}'}{x_k'} \right|
\le
2 + \frac{\sqrt\mu}{2\sin^2\theta_b}
\end{equation}
on $B_{k+2}$. Using the trace recursion \eqref{eq:trace:renorm} and its derivatives, we obtain
$$
\frac{x_{k+2}'}{x_k'}
=
4 x_{k+1} + 2x_{k-2} + (4 x_k^2 - 1)\frac{x_{k-1}'}{x_k'} - 2x_k \frac{x_{k-2}'}{x_k'}.
$$
This time, the $x_{k+1}$ term dominates. Using \eqref{eq:subcase2iderivbound}, the inductive hypothesis, and $B_{k+2} \subseteq B_k \subseteq \sigma_{k-2}$ to estimate the remaining terms for $z \in B_{k+2}$, we have
\begin{align*}
\left| 2x_{k-2} + (4 x_k^2 - 1)\frac{x_{k-1}'}{x_k'} - 2x_k \frac{x_{k-2}'}{x_k'} \right|
& \le
2 + 3\left(2+\frac{\sqrt\mu}{2\sin^2\theta_b} \right) + 2 \\
& =
10 + \frac{3\sqrt\mu}{2\sin^2\theta_b}
\end{align*}
on $B_{k+2}$. Using the triangle inequality again, we have
$$
\left( 4 - \frac{3}{2\sin^2\theta_b} \right) \sqrt\mu - 14
\le
\left| \frac{x_{k+2}'}{x_k'} \right|
\le
\left( 4 + \frac{3}{2\sin^2\theta_b} \right) \sqrt\mu + 18
$$
on $B_{k+2}$.

\bigskip

\noindent \textbf{Subcase 2ii: $B_k$ is type A.} We have $B_{k+2} \subseteq B_k \subseteq \sigma_{k-1}$, and $B_k \cap \sigma_{k-2} = \emptyset$. Differentiating the renormalization equation, we get
$$
\frac{x_{k+2}'}{x_k'}
=
2x_{k+1} + 2x_k \frac{x_{k+1}'}{x_k'} - \frac{x_{k-1}'}{x_k'}.
$$
Again, $\sqrt\mu - 1 \le |x_{k+1}| \le \sqrt\mu + 2$ on $B_{k+2}$, and, like the previous cases, we have $x_{k+1} = g_\pm(x_k, x_{k-1}, I)$, so
$$
\left|\frac{x_{k+1}'}{x_k'}\right|
\le
2 + \frac{\sqrt\mu}{2\sin^2\theta_b}
$$
on $B_{k+2}$. Therefore, we have
$$
\left| 2x_k \frac{x_{k+1}'}{x_k'} + \frac{x_{k-1}'}{x_k'} \right|
\le
5 + \frac{\sqrt\mu}{\sin^2\theta_b}
$$
on $B_{k+2}$. Applying the triangle inequality two more times, we have
$$
\left(2 - \frac{1}{\sin^2\theta_b} \right) \sqrt{\mu} - 7
\le
\left|\frac{x_{k+2}'}{x_k'}\right|
\le
\left(2 + \frac{1}{\sin^2\theta_b}\right)\sqrt{\mu} + 9
$$
on $B_{k+2}$.
\end{proof}

\begin{proof}[Proof of Lemma~\ref{l:dtkklbounds}]
This follows immediately from Lemma~\ref{l:tracederivbounds}.
\end{proof}

Lastly, we will prove the power-law bounds on transfer matrices at spectral parameters in the periodic approximating spectra. To do this, we need trace bounds on these spectra, which the following lemma supplies. The key ingredient is that the escape condition is a bit simpler to check in our situation. Notice that something must be said in our situation, since some choices of $\theta_a$ and $\theta_b$ can lead to $|x_{-1}(z)| > 1$ for $z \in \Sigma$, whereas one has $x_{-1} \equiv 1$ in the Schr\"odinger case.

\begin{lemma} \label{l:easyesc}
If $0 < \theta_b < \theta_a < \pi/2$, $\mu \ge 32$, and $z \in \partial \D$, the trace orbit $(x_k(z))_{k\ge-1}$ is unbounded if and only if there exists $k_0 \ge 0$ such that
\begin{equation} \label{eq:easyescape}
|x_{k_0}(z)|, \, |x_{k_0 + 1}(z)|
>
1.
\end{equation}
Consequently, if $z \in \Sigma = B_\infty$, one has $|x_k(z)| \le \kappa + 2$ for all $n \ge -1$, where $\kappa$ is as in \eqref{eq:fqw:kappadef}.
\end{lemma}

\begin{proof} By Lemma~\ref{l:bandtypes}, we have $\sigma_{-1} \supseteq \sigma_k$ for all $k$, and $\sigma_{-1} \supseteq \Sigma$. There are two possibilities: if $z \notin \sigma_{-1}$, then one has $z \notin \Sigma$ and \eqref{eq:easyescape} holds for some $k_0$. On the other hand, if $z \in \sigma_{-1}$, then $|x_{-1}(z)| \le 1$, and so the existence of $k_0$ which satisfies \eqref{eq:easyescape} is equivalent to the existence of $k_0$ which satisfies the general escape condition \eqref{eq:escape:cond}.

The claimed bound on $|x_k(z)|$ for $z \in B_\infty$ follows from a simple observation. Namely, if $|x_k(z)| > 1$ for some $z \in B_\infty$, then $|x_{k\pm 1}(z)| \le 1$, whence
$$
|x_k|
=
\left| x_{k-1}x_{k+1} \pm
\sqrt{I + (1-x_{k-1}^2)(1-x_{k+1}^2)} \right|
\le
1 + \sqrt{\kappa^2 + 1}
\le
2 + \kappa,
$$
where we have suppressed the dependence of the $x$'s and $I$ on $z$.

\end{proof}

\begin{proof}[Proof of Lemma~\ref{l:dt:it}]

Since this is nearly identical to \cite{IT91} and \cite[Lemma~3.2]{DT03}, we only sketch the main steps of the argument. The main idea is to begin by estimating the norms of the matrices $M_n$, and then to decompose an arbitrary $Z$-type matrix over an even number of sites into a product of $M_j$'s by using the fact that every integer has a natural ``base Fibonacci'' expansion. Finally, one can interpolate to get transfer matrices over an odd number of sites.

First, from Lemma~\ref{l:easyesc}, we have $|x_k(z)| \le \kappa + 2$ for every $z \in \sigma_k$, where $\kappa$ is as in \eqref{eq:fqw:kappadef}. As a first crude estimate, we have
\begin{equation}
\| M_j(z) \|
\le
(\kappa + 2)^j
\end{equation}
for $j \le k$ and $z \in \sigma_k$. As in \cite{DT03}, this follows from the recursive relationship among the $M_n$-type matrices and the Cayley--Hamilton Theorem:
$$
M_j
=
x_{j-1} M_{j-2} - M_{j-3}^{-1}.
$$
One also uses the fact that $\det M_\ell = 1$ for all $\ell$, which implies $\|M_{j-3}^{-1}\|  = \|M_{j-3}\|$.
\bigskip

As in \cite[Lemma~5]{IT91} and \cite[Proposition~3.2]{DT03}, for every $i \ge 2$, $j \ge 1$, we may construct polynomials $P_j^{(1)}, \ldots, P_j^{(4)}$ of degree at most $j$ in $x_{i-1}, \ldots,x_j$ such that
$$
M_k M_{k+j}
=
  P_j^{(1)} M_{i+j}
+ P_j^{(2)} M_{i+j-1}
+ P_j^{(3)} M_{i+j-2}
+ P_j^{(4)} I.
$$
Moreover,
\begin{equation} \label{eq:polygrowth1}
\sum_{\ell = 1}^4 \left| P_j^{(\ell)} \right| \big( |x_{i-1}|, \ldots, |x_{i+j}| \big)
\le
(2\kappa + 5)^j,
\end{equation}
where $|P|$ denotes the polynomial obtained by replacing each coefficient of $P$ by its absolute value. Indeed, \eqref{eq:polygrowth1} follows from a verbatim repition of the arguments of \cite{DT03, IT91}, since all that is used in those references is Cayley--Hamilton and the trace bounds for indices bounded above by $k$. By repeating the next step of the argument, we obtain
$$
\| Z(2m, 0 ; z) \|
\le
d^{b_m}
\text{ for all } 1 \le m \le F_k,
$$
where $d=(\kappa+2)(2\kappa+5)^2$ and $b_m = \max\{ n : F_n \le m \}$; in particular, $b_m \le k$. Since $F_n \sim \varphi^n$, we have
$$
b_m
\le
\frac{\log m}{\log\varphi} + D
$$
for some constant $D$, whence
$$
\| Z(2m,0;z) \|
\le
d^{\frac{\log m}{\log\varphi} + D}
=
d^D m^{\tau/2}.
$$
We obtain \eqref{eq:dt:it} for all $n > 0$ by interpolating to get odd values of $n$. The negative case follows from the reflection symmetry \eqref{eq:fib:refl}. One may also obtain the same the result for $T$-type matrices by using \eqref{eq:szego:gz:rel}.
\end{proof}

\section{The Thue-Morse Quantum Walk}\label{sec.8}

Let us consider a quantum walk as in Section~\ref{sec:fqw}, but for which the coins $R_a$ and $R_b$ are distributed according to an element of the Thue-Morse subshift, rather than the Fibonacci subshift.

More precisely, consider the alphabet $\A = \{a,b\}$ as before. The Thue-Morse substitution is defined by
$$
S:a \mapsto ab, \quad b \mapsto ba.
$$
Iterating $S$ on $a$, we obtain a sequence of words $w_n = S^n(a)$, and an infinite word
$$
u_{\TM}
=
\lim_{n \to \infty} w_n
=
abbabaabbaababba\ldots,
$$
and we define the Thue-Morse subshift to be the space of all sequences over the alphabet $\A$ with the same local structure as $u_{\TM}$, i.e.,
$$
\Omega_{\TM}
=
\left\{
\omega \in \A^{\Z} : \text{every finite subword of $\omega$ occurs in } u_{\TM}
\right\}.
$$
As before, we can associate to each $\omega \in \Omega_\TM$ a sequence of coins $Q_{\omega,n} = R_{\omega_n}$ and a sequence of Verblunsky coefficients $\alpha_\omega$ as in Section~\ref{sec:fqw}. Now, consider the transfer matrices over the words $w_n$. More precisely, fix $\alpha = \alpha_{\omega_0}$, where $\omega_0 = u_\TM^R | u_\TM \in \Omega_\TM$, and define
$$
M_n(z)
=
Z(2^{n+1}, 0;z),
\quad
n \geq 0, \, z \in \C.
$$
Also, define $M_{-1}$ as in the previous section. Let
$$
t_n(z)
=
\tr(M_n(z)).
$$
The recursive relationship between the $M_n$'s implies a recursive relationship amongst the $t_n$'s, viz.
\begin{equation} \label{eq:tmtracemap}
t_{n+1}
=
t_{n-1}^2 (t_n - 2) + 2, \, n \ge 2.
\end{equation}
The sets
$$
\sigma_n
=
\{ z \in \partial \D : |t_n(z)| \leq 2 \}
$$
correspond to spectra of quantum walk update rules whose coins are distributed $2^n$-periodicially according to $w_n$. Since the initial conditions are the same as in Section~\ref{sec:fqw}, we still have
\begin{align*}
t_{-1}(z)
& =
\mathrm{Re}(z) \sec\theta_b \\
t_0(z)
& =
\mathrm{Re}(z) \sec\theta_a \\
t_1(z)
& =
\mathrm{Re}\left(z^2\right) \sec\theta_a\sec\theta_b + \tan\theta_a \tan\theta_b
\end{align*}

\begin{lemma}\label{l:degenerate:ens}
If $z_0 \in \partial \D$ satisfies $t_n(z_0) = 2$ for some $n \ge 3$ and $t_2(z_0) \neq 2$, then $z_0$ corresponds to a closed gap of $\sigma_n$.
\end{lemma}

\begin{proof}
The proof is identical to the proof of the corresponding fact for Schr\"odinger/Jacobi transfer matrices, so we only sketch the details; compare \cite[Proposition~5.1]{DT03}. Using the trace recurstion \eqref{eq:tmtracemap}, the stated conditions on $z_0$ imply that $t_k(z_0) = 0$ for some $1 \leq k \leq n - 2$. One can then use this to show that $M_j(z_0) = I$ for all $j \ge k+1$ (in particular, for $j = n$). Alternatively, one can differentiate \eqref{eq:tmtracemap} to see that $t_j'(z_0) = 0$ for all $j \ge k+1$.
\end{proof}

\begin{theorem}\label{t:tmqw:transp}
For all $\omega \in \Omega_{\TM}$, one has $\widetilde\beta^-_{\delta_{-1}} (p)
\geq
1 - \frac{1}{p}$.
\end{theorem}

\begin{proof}
This follows from Lemma~\ref{l:degenerate:ens} and Theorem~\ref{t:dt:powerlaw}, just as in the proof of Theorem~\ref{t:poly:deloc}. In particular, we take $A_R = \{z_0\}$ for all $R \ge 1$, where $z_0 \in \partial \D$ satisifes $t_2(z_0) \neq 2$ and $t_n(z_0) = 2$ for some $n \ge 3$; compare \cite[Theorem~4]{DT03}.
\end{proof}

\section{Ballistic Wavepacket Propagation for Periodic CMV Matrices}\label{sec.9}

We will say that a CMV matrix $\E$ is $p$-periodic if $\alpha_{n+p} = \alpha_n$ for all $n \in \Z$. If necessary, double $p$ so that it can be assumed to be even throughout this section, which is no problem, since nothing that we will do requires a minimal period. This can be thought of as a crystalline model, so the physical expectation is that resistance should go to zero and transport of wave packets should be free (ballistic). The following theorem establishes such a result.\footnote{As we mentioned in the introduction, it is possible to derive this result from \cite[Theorem~4]{AVWW}. We are grateful to Albert H.\ Werner for pointing this out to us.}

\begin{theorem}\label{t:ballistic}
Suppose $\E$ is $p$-periodic for some $p \in \Z_+$. Then, for any $\psi \in \mathcal S$, $\beta^{\pm}_\psi(q) = \widetilde \beta_\psi^\pm(q) = 1$ for all $q > 0$.
\end{theorem}

It is immediately clear that Theorem~\ref{t:ballistic} generalizes to CMV matrices whose Verblunsky coefficients are \emph{skew-periodic} in the sense that $\alpha_{n+p} =  \omega \alpha_n$ for some $\omega$ with $|\omega|=1$ and all $n \in \Z$. This also implies that a quantum walk with periodically distributed coins will exhibit ballistic transport.

\begin{coro}\label{coro:skewper}
Suppose $\E$ is skew-periodic. Then, for any $\psi \in \mathcal S$, $\beta^{\pm}_\psi(q) = \widetilde \beta_\psi^\pm(q) = 1$ for all $q > 0$.
\end{coro}

\begin{coro}\label{coro:perqw}
Suppose $U$ is the update rule of a quantum walk whose coins obey $Q_{n+p} = Q_n$ for all $n \in \Z$. Then, for any $\psi \in \mathcal S$, $\beta^{\pm}_\psi(q) = \widetilde \beta_\psi^\pm(q) = 1$ for all $q > 0$.
\end{coro}

\begin{proof}[Proof of Corollary~\ref{coro:skewper}]
Let $\omega = e^{ip\phi}$ for some $\phi \in \R$, and introduce $\widetilde \E$ with coefficients $\widetilde \alpha$ defined by
$$
\widetilde \alpha_j
=
e^{-ij\phi} \alpha_j,
\quad
j \in \Z.
$$
It is not hard to see that $\widetilde \alpha$ is honestly $p$-periodic, so we may apply Theorem~\ref{t:ballistic} to deduce ballistic transport for $\widetilde \E$. Moreover, if we define the diagonal unitary $\Gamma$ on $\ell^2$ by
$$
\gamma_k
=
\exp\left( i (-1)^{k+1} \left\lfloor \frac{k+1}{2} \right\rfloor \phi\right),
\quad
\langle \delta_j, \Gamma \delta_k \rangle
=
\delta_{j,k} \gamma_k,
$$
then $\E = e^{-i \phi} \Gamma^* \widetilde \E \Gamma$. Thus, we deduce ballistic transport for $\E$.
\end{proof}

\begin{proof}[Proof of Corollary~\ref{coro:perqw}] Let $\Lambda$ and $\E = \E_\alpha$ be as in Section~\ref{sec:qw} so that $\E = \Lambda^* U \Lambda$. From the explicit form of $\Lambda$, we see that it maps $\mathcal S(\ell^2(\Z))$ onto $\mathcal S(\ell^2(\Z) \otimes \C^2)$, and it does not affect transport exponents. Thus, it suffices to deduce ballistic transport for $\E$. Examining \eqref{e.correspondence}, we see that $\alpha_{n+2p} = \omega \alpha_n$ for all $n \in \Z$, where $\omega$ is given by
$$
\overline{\omega}
=
\prod_{j=1}^p \omega_j^1 \omega_j^2.
$$
Thus, $\E$ is skew-periodic, so ballistic transport for $U$ follows from Corollary~\ref{coro:skewper}.
\end{proof}

In order to apply the methods of \cite{AK98} and \cite{DLY} to deduce ballistic transport, we need to establish effective estimates on the Heisenberg evolution of the position operator, $X$, so we want to consider $\E X - X\E$ as in the other cases. We will actually work with a 2-block version of $X$, as this substantially simplifies the resulting commutators, without affecting the transport exponents.

To set up notation, let $\ell^0 =\ell^0(\Z)$ denote the space of finitely supported sequences, and define a 2-block variant of the position operator $X:\ell^0 \to \ell^0$ by
$$
X\delta_n
=
\underline n \delta_n
:=
2 \left\lfloor \frac{n}{2} \right\rfloor \delta_n
$$
for $n \in \Z$. Clearly $X$ is essentially self-adjoint on $\ell^0$. For any observable (read:\ operator) $A$, denote its discrete Heisenberg evolution with respect to $\E$ by
$$
A(j)
=
\E^j A \E^{-j},
\quad
j \in \Z.
$$
Obviously, $\ell^0$ is invariant under both $\E$ and $X$, so that $X(j)$ is a well-defined symmetric operator which is essentially self-adjoint on $\ell^0$ for all $j \in \Z$. As in \cite{DLY}, the key ingredient is furnished by a variation on the theme of Asch-Knauf \cite{AK98}.

\begin{theorem} \label{t:aschknauf}
There is a bounded, self-adjoint operator $J$ such that
$$
\slim_{L \to \infty}
\frac{1}{L} X(L) = J.
$$
Moreover, $\ker J = \{0\}$, and
$$
\slim_{L \to \infty} f\left( \frac{1}{L} X(L) \right)
=
f(J)
$$
for any bounded continuous function $f$.
\end{theorem}

\begin{remark} Notice that Theorem~\ref{t:aschknauf} does not follow from \cite[Theorem~1.6(a),(b)]{DLY} in any straightforward fashion. In order to deduce our theorem from theirs, one would need to be able to write $\E = e^{iH}$ with $H$ a Jacobi matrix (or some other self-adjoint operator of finite width), and it is not at all obvious that this is the case. Perhaps more importantly, even if such a proof could be found, it would add no insight to the unitary setting, whereas our proof explicitly constructs $J$ out of computable functions of $\E$.
\end{remark}

It is important to point out that the result is not trivial, even though we follow the outline of \cite{DLY}. In particular, one needs an identity like \eqref{eq:DLYmagic} to make the argument work, and the more complicated structure of CMV matrices makes such an identity more difficult to prove.

\subsection{Direct Integral Decomposition of Periodic CMV Matrices}

The key ingredient here is to decompose a periodic CMV operator as a direct integral of unitary operators on $\C^p$. Since this is standard fare, we only summarize the theorems and provide no proofs. See \cite[Section~11.2]{S2} for more details on periodic CMV operators; see also \cite[Chapter~5]{simszego} for the Jacobi version.

Given a sequence $\alpha$ of Verblunsky coefficients, for notational convenience, introduce $a,b,c,$ and $d$ as in \cite{DFLY2}, i.e.,
\begin{align*}
a_n = -\overline{\alpha_n} \alpha_{n-1}, \quad
b_n = \overline{\alpha_n} \rho_{n-1}, \quad
c_n = -\rho_n \alpha_{n-1}, \quad
d_n = \rho_n \rho_{n-1},
\quad
n \in \Z.
\end{align*}
In terms of these parameters, the matrix representation of $\E$ becomes
$$
\E
=
\begin{pmatrix}
\ddots & \ddots & \ddots &&&&&  \\
 & a_0 & b_1 & d_1 &&& & \\
& c_0 & a_1 & c_1 &&& & \\
&  & b_2 & a_2 & b_3 & d_3 & & \\
& & d_2 & c_2 & a_3 & c_3 &  &  \\
& &&& b_4 & a_4 & b_5 & \\
& &&& d_4 & c_4 & a_5 &   \\
& &&&& \ddots & \ddots &  \ddots
\end{pmatrix}.
$$
If $\E$ is $p$-periodic, then a mod-$p$ variant of the Fourier transform diagonalizes $\E$ into a multiplication operator on $\Hi_p = L^2\left([0,2\pi), \C^p ; \frac{d\theta}{2\pi} \right)$, the space of functions $f:[0,2\pi) \to \C^p$ which are square integrable with respect to Lebesgue measure. More specifically, for each $\theta$, denote
$$
\E_\theta
=
\begin{pmatrix}
a_1 & c_1 &&&&&& d_p e^{-i\theta} & c_p e^{-i\theta} \\
b_2 & a_2 & b_3 & d_3 &&&&& \\
d_2 & c_2 & a_3 & c_3 & \ddots &&&& \\
&& b_4 & a_4 & \ddots & b_{p-3} & d_{p-3} && \\
&& d_4 & c_4 & \ddots & a_{p-3} & c_{p-3} && \\
&&&& \ddots & b_{p-2} & a_{p-2} & b_{p-1} & d_{p-1} \\
&&&&& d_{p-2} & c_{p-2} & a_{p-1} & c_{p-1} \\
b_1 e^{i\theta} & d_1 e^{i\theta} &&&&&& b_p & a_p
\end{pmatrix},
$$
and define the mod-$p$ Fourier transform $\mathcal F : \ell^2(\Z) \to \mathcal H_p$ by $u \mapsto \widehat u$, where
$$
\widehat u_j(\theta)
=
\sum_{\ell \in \Z} u_{j + \ell p} e^{-i \ell \theta},
\quad
1 \leq j \leq p, \, \theta \in [0,2\pi).
$$
Notice that our definition of $\E_\theta$ has already implicitly used evenness of $p$; we will not comment on the parity of $p$ again. The inverse of $\mathcal F$ is given by $\mathcal H_p \ni f \mapsto \check f \in\ell^2(\Z)$, where
$$
\check f_{j + \ell p}
=
\int_0^{2\pi} \! f_j(\phi) e^{i \ell \phi} \, \frac{d\phi}{2\pi}.
$$
For any operator $A$ on $\ell^2(\Z)$, we denote by $\widehat A$ its action on Fourier space, i.e., $\widehat A = \mathcal F A \mathcal F^{-1}$. After simple calculations, we see that $\left( \widehat{\mathcal\E}f \right)(\theta) = \E_\theta f(\theta)$ for each $f \in \Hi_p$, and almost every $\theta \in [0,2\pi)$. It is typical (especially in the physics literature) to view the Hilbert space $\Hi_p$ and the matrix $\widehat\E$ as direct integrals of finite-dimensional objects:
$$
\mathcal H_p
\cong
\int_{[0,2\pi)}^\oplus \! \C^p \, \frac{d\theta}{2\pi},
\quad
\widehat{\mathcal\E}
\cong
\int_{[0,2\pi)}^\oplus \! \E_\theta \, \frac{d\theta}{2\pi}.
$$

\subsection{Proof of Ballistic Transport}

The proofs of Theorem~\ref{t:ballistic} and \ref{t:aschknauf} are similar to those in \cite{DLY}, so we will only comment in detail about the challenges which one must overcome to rerun the Asch-Knauf-Damanik-Lukic-Yessen machine in the present setting, and we will be somewhat breezy with those details which are similar to the proofs in \cite{DLY}. The main difficulty here is that time is discrete in our case, so we must use discrete derivatives with respect to time parameters, and this causes some minor headaches.

Formally, let $B = [\E,X] = \E X - X \E$. More precisely, $B$ is initially defined on $\ell^0$, and then it is obviously a bounded operator thereupon (since $\E$ is five-diagonal), and thus, $B$ enjoys a unique extension to a bounded operator on $\ell^2(\Z)$ by general nonsense. Now, observe that the Heisenberg evolution of $B$ is (almost) the discrete derivative of the Heisenberg evolution of $X$:
\begin{equation} \label{eq:heis:int}
X(k)
=
X + \left( \sum_{\ell=0}^{k-1} B(\ell) \right)\E^*.
\end{equation}
The post-factor of $\E^*$, though unsightly, is ultimately irrelevant; more specifically, strong convergence of $\langle B \rangle(L)$  clearly implies strong convergence of $\langle B \rangle(L) \cdot \E^*$ as $L \to \infty$. A short calculation reveals
$$
B
=
\begin{pmatrix}
\ddots & \ddots & \ddots &&&&&  \\
& 0 & 0 & 2d_1 &&&& \\
& 0 & 0 & 2c_1 &&&& \\
&& -2b_2 & 0 & 0 & 2d_3 && \\
&& -2d_2 & 0 & 0 & 2c_3 && \\
&&&& -2b_4 & 0 & 0 & \\
&&&& -2d_4 & 0 & 0 & \\
&&&&& \ddots & \ddots & \ddots
\end{pmatrix}.
$$
Since $B$ is also clearly $p$-periodic, $\mathcal F$ also diagonalizes it into a direct integral as well. Specifically, we have $\left( \widehat B f \right)(\theta) = B_\theta f(\theta)$ for all $f \in \mathcal H_p$ and almost every $\theta \in [0,2\pi)$, where
\begin{align*}
B_\theta
=
\begin{pmatrix}
0 & 2c_1 &&&& -2d_p e^{-i\theta} &  0 \\
-2b_2 & 0 & 0 & 2d_3 &&& \\
-2d_2 & 0 & 0 & 2c_3 &&& \\
&& \ddots & \ddots & \ddots && \\
&&& -2b_{p-2} & 0 & 0 & 2d_{p-1} \\
&&& -2d_{p-2} & 0 & 0 & 2c_{p-1} \\
0 & 2d_1 e^{i\theta} &&&& -2b_p & 0 \\
\end{pmatrix}
\end{align*}
Now let $V_\theta \in \C^{p \times p}$ denote the diagonal matrix with diagonal entries $\langle e_k, V_\theta e_k \rangle = e^{i\underline{k}\theta/p} = e^{2i\lfloor k/2 \rfloor \theta/p}$, i.e.,
$$
V_\theta
=
\begin{pmatrix}
1 &&&&&& \\
& e^{2i\theta/p} &&&&& \\
&& e^{2i\theta/p} &&&& \\
&&& \ddots &&& \\
&&&& e^{i\theta(p-2)/p} && \\
&&&&& e^{i\theta(p-2)/p} & \\
&&&&&& e^{i\theta}
\end{pmatrix}.
$$
Define
$$
\widetilde \E_\theta
=
V_\theta^{-1} \E_\theta V_\theta,
\quad
\widetilde B_\theta
=
V_\theta^{-1} B_\theta V_\theta
$$
By a direct calculation, one can check that
\begin{equation} \label{eq:DLYmagic}
\frac{\partial}{\partial \theta} \widetilde \E_\theta
=
\frac{i}{p} \widetilde B_\theta.
\end{equation}
With this setup, we run the argument from \cite{DLY}, which gives us
\begin{equation} \label{eq:DLY9}
\lim_{L \to \infty} \frac{1}{L} \sum_{\ell = 0}^{L-1} \E_\theta^\ell B_\theta \E_\theta^{-\ell}
=
-ip \sum_{j=1}^m \frac{\partial \lambda_j}{\partial \theta}(\theta) P_j(\theta)
\end{equation}
for all $\theta \in (0,2\pi) \setminus \{\pi\}$, where $\lambda_1(\theta) < \cdots < \lambda_m(\theta)$ are the distinct eigenvalues of $\E_\theta$ and $P_j(\theta)$ denotes projection onto the eigenspace corresponding to $\lambda_j(\theta)$. Notice that $P_j(\theta)$ is a one-dimensional projection unless $\theta = 0,\pi$.
Let us briefly recapitulate the argument from \cite{DLY} which proves \eqref{eq:DLY9}. First, let $v_j(\theta)$ be an analytic choice of a normalized element of $\mathrm{ran}(P_j(\theta)) = \ker(\E_\theta - \lambda_j(\theta))$.  As in \cite{DLY}, with $\widetilde v_j(\theta) = V_\theta^{-1} v_j(\theta$), we have
\begin{align*}
\frac{\partial \lambda_j}{\partial \theta}(\theta)
& =
\frac{\partial}{\partial \theta} \langle v_j(\theta), \E_\theta v_j(\theta) \rangle \\
& =
\frac{\partial}{\partial \theta} \langle \widetilde v_j(\theta), \widetilde\E_\theta \widetilde v_j(\theta) \rangle \\
& =
\frac{i}{p} \langle \widetilde v_j(\theta), \widetilde B_\theta \widetilde v_j(\theta) \rangle \\
& =
\frac{i}{p} \langle  v_j(\theta),  B_\theta  v_j(\theta) \rangle.
\end{align*}
Additionally,
\begin{align*}
\lim_{L \to \infty}
 \frac{1}{L} \left\langle v_j(\theta),  \sum_{\ell=0}^{L-1} \E_\theta^\ell B_\theta \E_\theta^{-\ell} v_k(\theta) \right\rangle
& =
\lim_{L\to \infty}
\frac{1}{L} \sum_{\ell=0}^{L-1} \frac{\lambda_j^\ell(\theta)}{\lambda_k^\ell(\theta)} \langle v_j(\theta), B_\theta v_k(\theta) \rangle \\
& =
\begin{cases}
\langle v_j(\theta), B_\theta v_k(\theta) \rangle & \text{ if } j = k \\
0 & \text{ otherwise}
\end{cases}
\end{align*}
Thus, \eqref{eq:DLY9} follows.

\bigskip

\begin{proof}[Proof of Theorem~\ref{t:aschknauf}] Once one has \eqref{eq:DLY9}, then we can simply follow the proof of \cite[Theorem~1.6]{DLY}. In particular, we get the  conclusions of Theorem~\ref{t:aschknauf} with
$$
J
=
-ip \mathcal F^{-1} \left( \int_{[0,2\pi)}^\oplus \! \sum_{j=1}^p \frac{\partial \lambda_j}{\partial \theta}(\theta) P_j(\theta) \, \frac{d\theta}{2\pi} \right) \mathcal F \E^*.
$$
More specifically, dominated convergence yields
\begin{equation} \label{eq:aschknaufconv}
\lim_{L \to \infty} \int_{[0,2\pi)}^\oplus \left( \frac{1}{L} \sum_{\ell = 0}^{L-1} \E_\theta^\ell B_\theta \E_\theta^{-\ell} \right) \, \frac{d\theta}{2\pi}
=
-ip\int_{[0,2\pi)}^\oplus \sum_{j=1}^m \frac{\partial \lambda_j}{\partial \theta}(\theta) P_j(\theta) \, \frac{d\theta}{2\pi},
\end{equation}
where the limit on the left hand side is taken in the operator norm topology. The left hand side is simply
$$
\lim_{L\to\infty}\frac{1}{L} \sum_{\ell=0}^{L-1} \widehat{B(\ell)},
$$
so, we can conjugate both sides of \eqref{eq:aschknaufconv} by $\mathcal F$ to obtain
$$
\lim_{L \to \infty} \frac{1}{L} \sum_{\ell = 0}^{L-1} B(\ell)
=
J\E
$$
The theorem then follows from \eqref{eq:heis:int} and the arguments in \cite{DLY}.
\end{proof}

\begin{proof}[Proof of Theorem~\ref{t:ballistic}]
This follows from Theorem~\ref{t:aschknauf} just as in \cite{DLY}, since our choice of $X$ does not change the transport exponents.
\end{proof}

\end{document}